\documentclass[12pt]{amsart}
\usepackage{amsmath,amssymb,amsfonts,amscd,epsf,graphicx}
\begin{document} 
\theoremstyle{plain}
\newtheorem{thm}{Theorem}
\newtheorem{lem}[thm]{Lemma}
\newtheorem{cor}[thm]{Corollary}
\newtheorem{prop}[thm]{Proposition}
\newtheorem{remark}[thm]{Remark}
\newtheorem{rmk}[thm]{Remark}
\newtheorem{defn}[thm]{Definition}
\newtheorem{ex}[thm]{Example}
\newtheorem{conj}[thm]{Conjecture}
%\numberwithin{equation}{section}
\newcommand{\eq}[2]{\begin{equation}\label{#1}#2 \end{equation}}
\newcommand{\ml}[2]{\begin{multline}\label{#1}#2 \end{multline}}
\newcommand{\ga}[2]{\begin{gather}\label{#1}#2 \end{gather}}
\newcommand{\mc}{\mathcal}
\newcommand{\mb}{\mathbb}
\newcommand{\surj}{\twoheadrightarrow}
\newcommand{\inj}{\hookrightarrow}
\newcommand{\red}{{\rm red}}
\newcommand{\codim}{{\rm codim}}
\newcommand{\rank}{{\rm rank}}
\newcommand{\Pic}{{\rm Pic}}
\newcommand{\Div}{{\rm Div}}
\newcommand{\Hom}{{\rm Hom}}
\newcommand{\im}{{\rm im}}
\newcommand{\Spec}{{\rm Spec \,}}
\newcommand{\Sing}{{\rm Sing}}
\newcommand{\Char}{{\rm char}}
\newcommand{\Tr}{{\rm Tr}}
\newcommand{\Gal}{{\rm Gal}}
\newcommand{\Min}{{\rm Min \ }}
\newcommand{\Max}{{\rm Max \ }}
\newcommand{\ti}{\times }
% Skriptbuchstaben
\newcommand{\sA}{{\mathcal A}}
\newcommand{\sB}{{\mathcal B}}
\newcommand{\sC}{{\mathcal C}}
\newcommand{\sD}{{\mathcal D}}
\newcommand{\sE}{{\mathcal E}}
\newcommand{\sF}{{\mathcal F}}
\newcommand{\sG}{{\mathcal G}}
\newcommand{\sH}{{\mathcal H}}
\newcommand{\sI}{{\mathcal I}}
\newcommand{\sJ}{{\mathcal J}}
\newcommand{\sK}{{\mathcal K}}
\newcommand{\sL}{{\mathcal L}}
\newcommand{\sM}{{\mathcal M}}
\newcommand{\sN}{{\mathcal N}}
\newcommand{\sO}{{\mathcal O}}
\newcommand{\sP}{{\mathcal P}}
\newcommand{\sQ}{{\mathcal Q}}
\newcommand{\sR}{{\mathcal R}}
\newcommand{\sS}{{\mathcal S}}
\newcommand{\sT}{{\mathcal T}}
\newcommand{\sU}{{\mathcal U}}
\newcommand{\sV}{{\mathcal V}}
\newcommand{\sW}{{\mathcal W}}
\newcommand{\sX}{{\mathcal X}}
\newcommand{\sY}{{\mathcal Y}}
\newcommand{\sZ}{{\mathcal Z}}
% Sonderbuchstaben mit Doppellinie
\newcommand{\A}{{\Bbb A}}
\newcommand{\B}{{\Bbb B}}
\newcommand{\C}{{\Bbb C}}
\newcommand{\D}{{\Bbb D}}
\newcommand{\E}{{\Bbb E}}
\newcommand{\F}{{\Bbb F}}
\newcommand{\G}{{\Bbb G}}
\renewcommand{\H}{{\Bbb H}}
\newcommand{\I}{{\Bbb I}}
\newcommand{\J}{{\Bbb J}}
\newcommand{\M}{{\Bbb M}}
\newcommand{\N}{{\Bbb N}}
\renewcommand{\P}{{\Bbb P}}
\newcommand{\Q}{{\Bbb Q}}
\newcommand{\R}{{\Bbb R}}
\newcommand{\T}{{\Bbb T}}
\newcommand{\U}{{\Bbb U}}
\newcommand{\V}{{\Bbb V}}
\newcommand{\W}{{\Bbb W}}
\newcommand{\X}{{\Bbb X}}
\newcommand{\Y}{{\Bbb Y}}
\newcommand{\Z}{{\Bbb Z}}
\newcommand{\pic}{{\text{Pic}(C,\sD)[E,\nabla]}}
\newcommand{\ocd}{{\Omega^1_C\{\sD\}}}
\newcommand{\oc}{{\Omega^1_C}}
\newcommand{\al}{{\alpha}}
\newcommand{\ta}{{\theta}}
\newcommand{\ve}{{\varepsilon}}
\newcommand{\lr}[2]{\langle #1,#2 \rangle}
\newcommand{\nnn}{\newline\newline\noindent}
\newcommand{\nn}{\newline\noindent}

\newcommand{\be}{\begin{equation}}
\newcommand{\ee}{\end{equation}}
\newcommand{\bea}{\begin{eqnarray}}
\newcommand{\eea}{\end{eqnarray}}
\newcommand{\beas}{\begin{eqnarray*}}
\newcommand{\eeas}{\end{eqnarray*}}

\def\One{\mathbb{I}}

\newcommand{\bl}{{\bf SB }}

\def\partdunce{{\;\raisebox{-40mm}{\epsfysize=80mm\epsfbox{partdunce.eps}}\;}}
\def\bintree{{\;\raisebox{-30mm}{\epsfysize=60mm\epsfbox{bintree.eps}}\;}}
\def\dunce{{\;\raisebox{-20mm}{\epsfysize=60mm\epsfbox{DunceParCut.eps}}\;}}
\def\OutT{{\;\raisebox{-20mm}{\epsfysize=60mm\epsfbox{OuterTriangle.eps}}\;}}
\def\Cell{{\;\raisebox{-30mm}{\epsfysize=75mm\epsfbox{Cell.eps}}\;}}
\def\VertT{{\;\raisebox{-20mm}{\epsfysize=60mm\epsfbox{VertTriangle.eps}}\;}}
\def\OutTwo{{\;\raisebox{-40mm}{\epsfysize=100mm\epsfbox{Out2.eps}}\;}}
\def\OutFeyn{{\;\raisebox{-20mm}{\epsfysize=50mm\epsfbox{OutFeyn.eps}}\;}}
\def\Spine{{\;\raisebox{-40mm}{\epsfysize=80mm\epsfbox{Spine.eps}}\;}}
\def\DunceVar{{\;\raisebox{-30mm}{\epsfysize=60mm\epsfbox{DunceVar.eps}}\;}}
\def\wthree{{\;\raisebox{-20mm}{\epsfysize=50mm\epsfbox{wthreemark.eps}}\;}}
\def\owt{{\;\raisebox{-30mm}{\epsfysize=70mm\epsfbox{outerweelthree.eps}}\;}}
\def\Gg{{\;\raisebox{-20mm}{\epsfysize=48mm\epsfbox{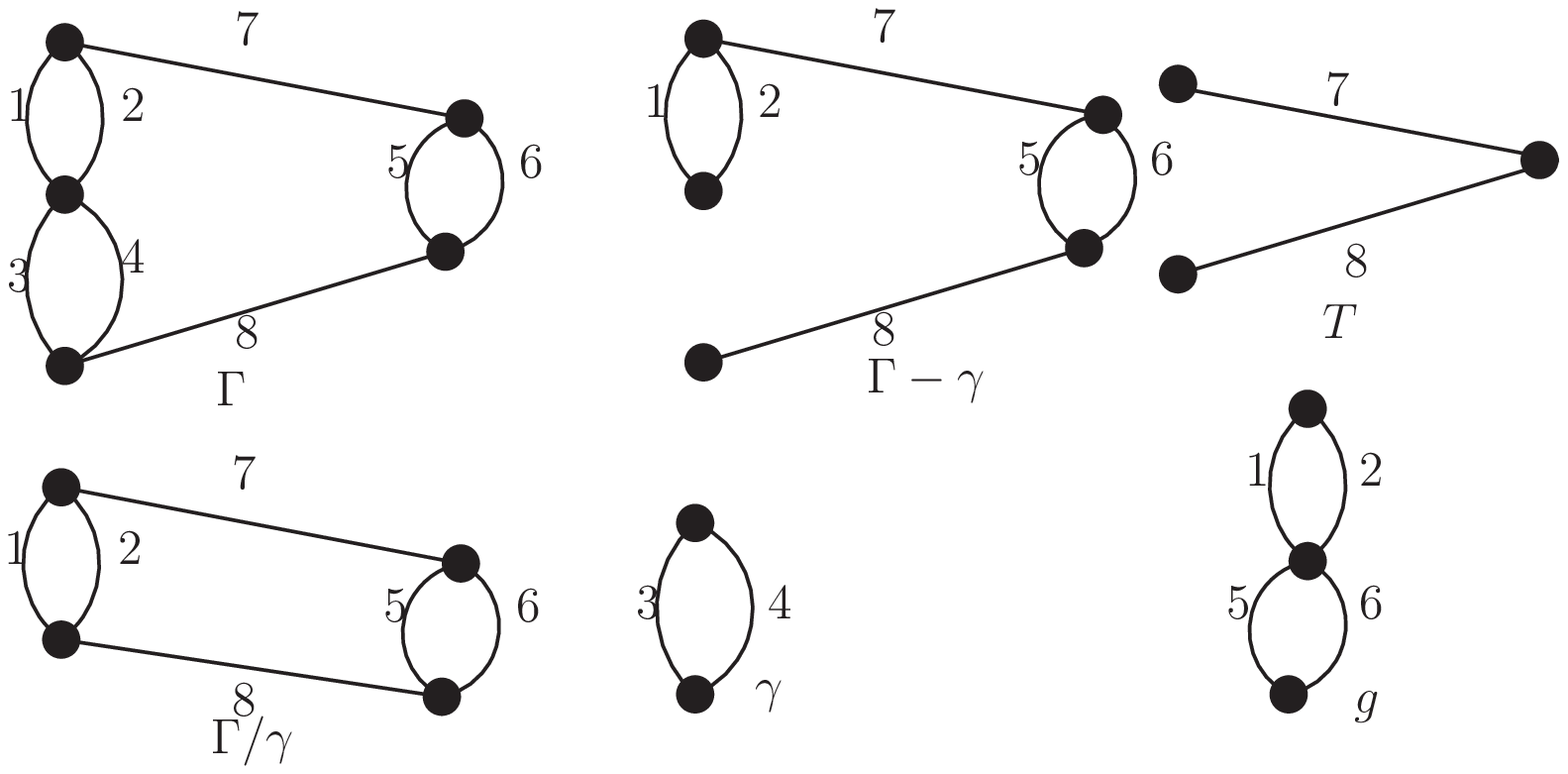}}\;}}
\def\ExCop{{\;\raisebox{-30mm}{\epsfysize=72mm\epsfbox{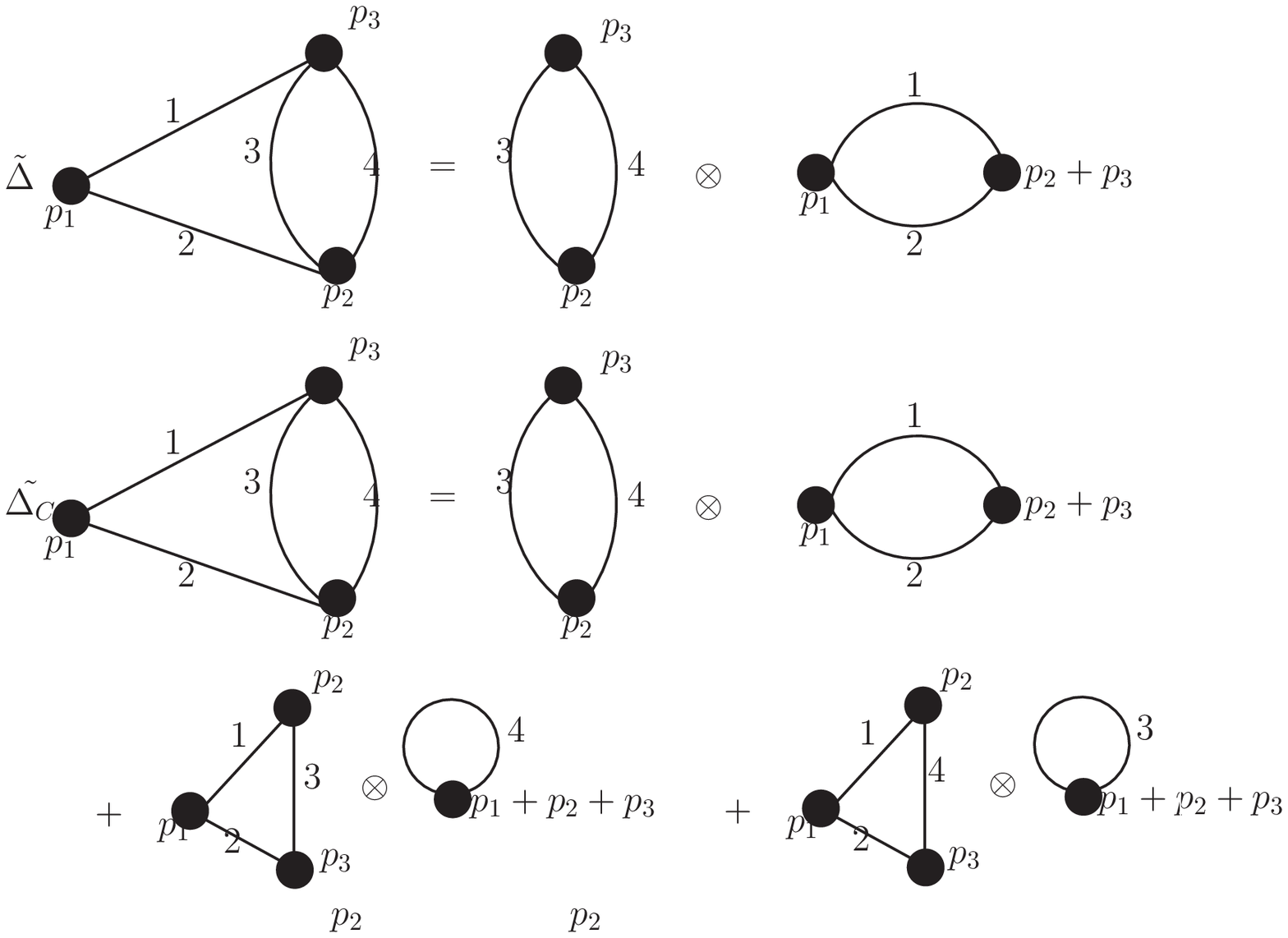}}\;}}
\def\OutSpBubble{{\;\raisebox{-40mm}{\epsfysize=100mm\epsfbox{OutSpBubble.eps}}\;}}
\def\OutSpTriangle{{\;\raisebox{-30mm}{\epsfysize=72mm\epsfbox{OutSpTriangle.eps}}\;}}
\def\OutSpDunce{{\;\raisebox{-30mm}{\epsfysize=72mm\epsfbox{OutSpDunce.eps}}\;}}
\def\triangle{{\;\raisebox{-30mm}{\epsfysize=80mm\epsfbox{triangle.eps}}\;}}
\def\ftwo{{\;\raisebox{30mm}{\epsfysize=70mm\epsfbox{ftwo.eps}}\;}}
\def\ftwolinear{{\;\raisebox{10mm}{\epsfysize=20mm\epsfbox{ftwolinear.eps}}\;}}
\def\fthree{{\;\raisebox{30mm}{\epsfysize=100mm\epsfbox{fthree.eps}}\;}}
\def\fthreetile{{\;\raisebox{30mm}{\epsfysize=90mm\epsfbox{FthreeTile.eps}}\;}}
\def\nondegg{{\;\raisebox{10mm}{\epsfysize=30mm\epsfbox{nondegg.eps}}\;}}
\def\Duncedetail{{\;\raisebox{35mm}{\epsfysize=80mm\epsfbox{Duncedetail.eps}}\;}}
\def\DunceA{{\;\raisebox{20mm}{\epsfysize=60mm\epsfbox{DunceA.eps}}\;}}
\def\DunceB{{\;\raisebox{20mm}{\epsfysize=60mm\epsfbox{DunceB.eps}}\;}}
\def\dcepstwo{{\;\raisebox{20mm}{\epsfysize=60mm\epsfbox{dcepstwo.eps}}\;}}
\def\CrossGeneral{{\;\raisebox{20mm}{\epsfysize=60mm\epsfbox{CrossingGeneral.eps}}\;}}
\def\CrossTriangle{{\;\raisebox{20mm}{\epsfysize=40mm\epsfbox{CrossingTriangle.eps}}\;}}
\def\graphcutshrink{{\;\raisebox{20mm}{\epsfysize=40mm\epsfbox{graphcutshrink.pdf}}\;}}
\def\cubical{{\;\raisebox{40mm}{\epsfysize=80mm\epsfbox{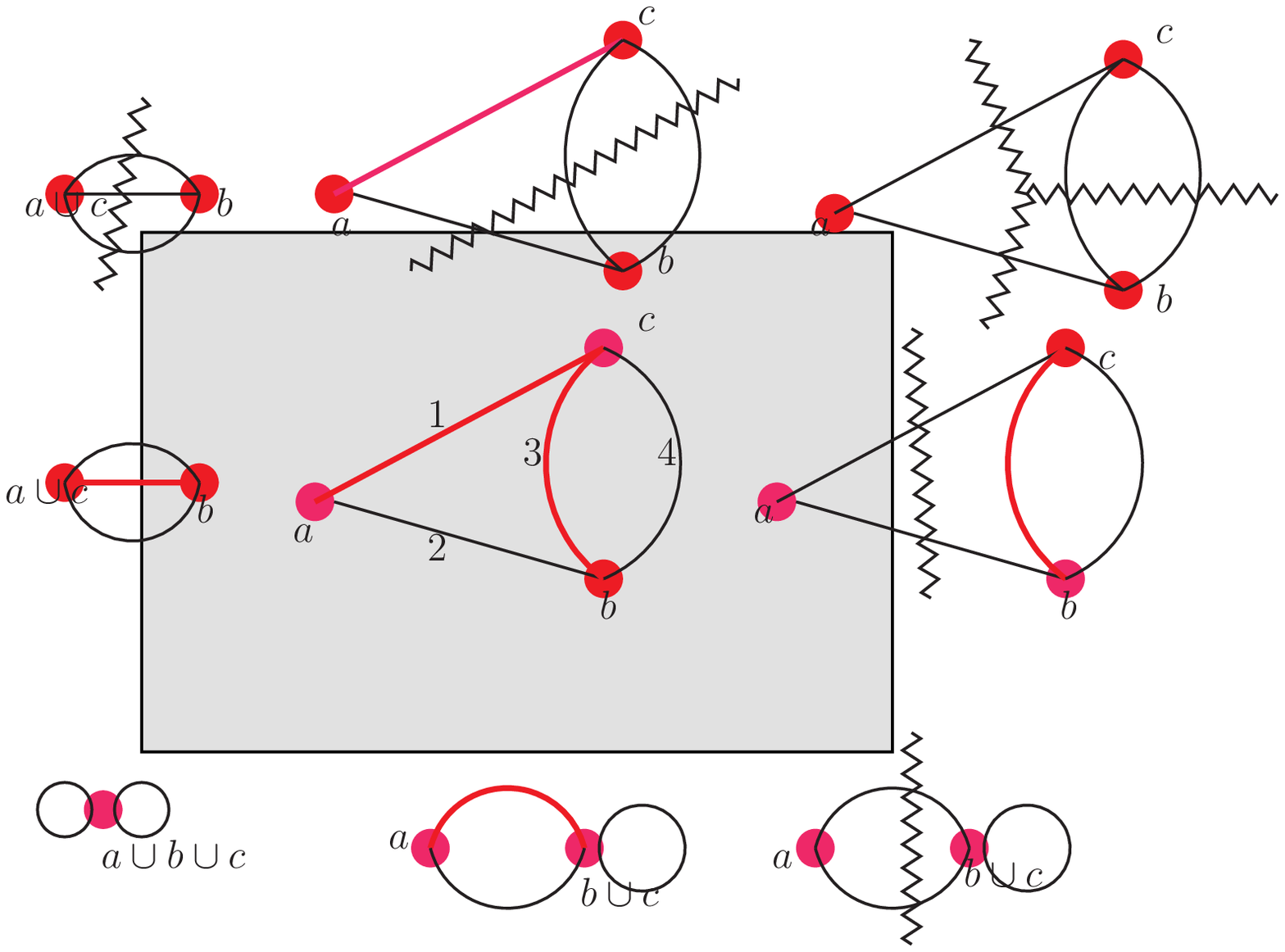}}\;}}
\def\cubicalM{{\;\raisebox{40mm}{\epsfxsize=140mm\epsfbox{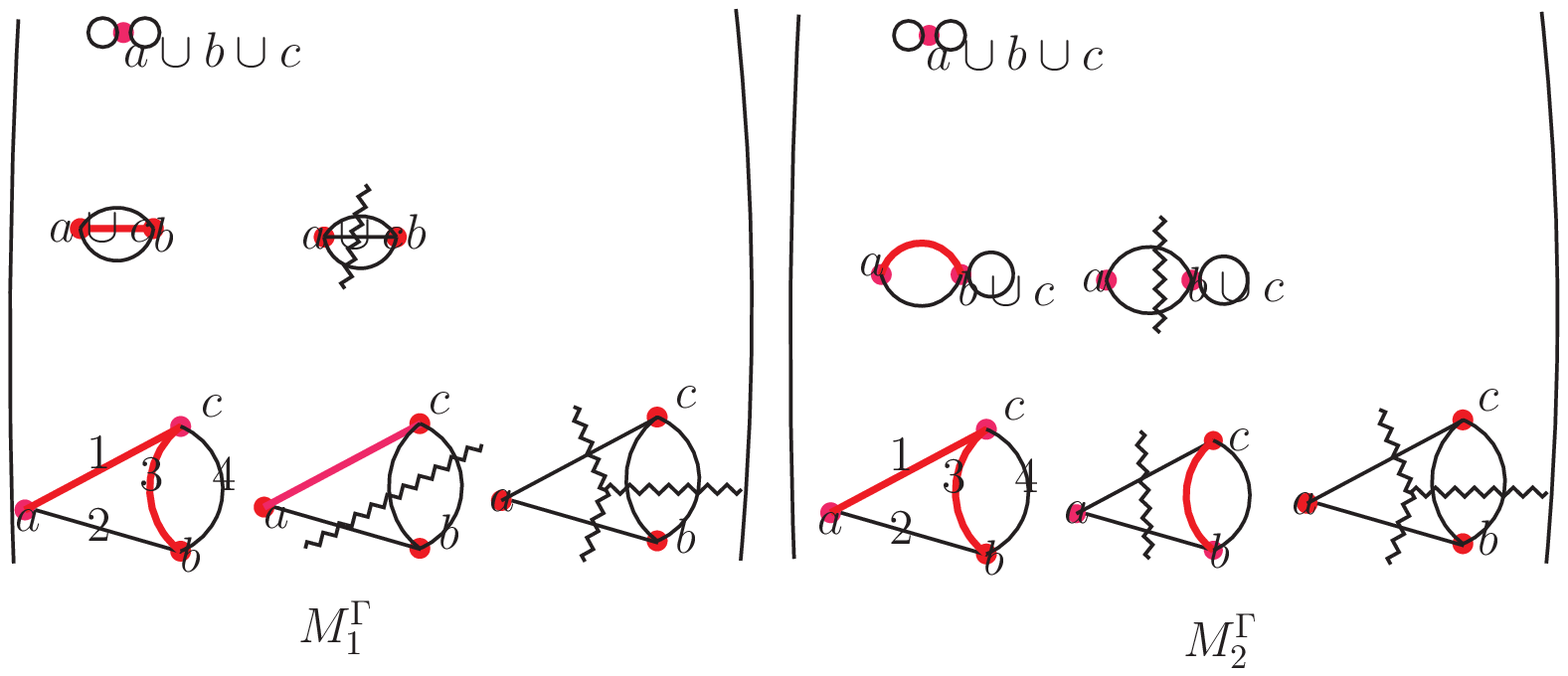}}\;}}
\def\bubblematrix{{\;\raisebox{-30mm}{\epsfxsize=60mm\epsfbox{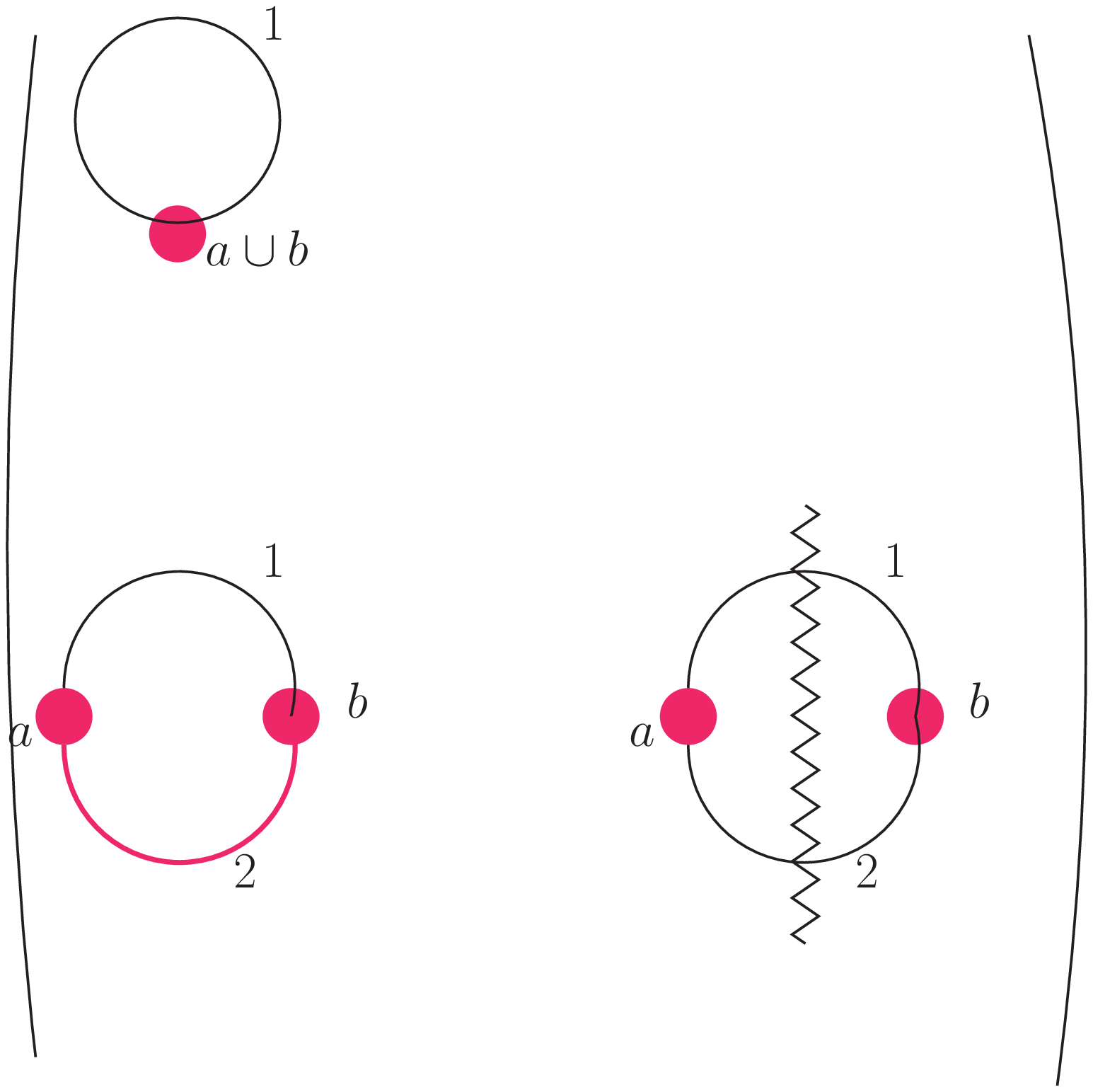}}\;}}
\def\trianglematrix{{\;\raisebox{-40mm}{\epsfxsize=80mm\epsfbox{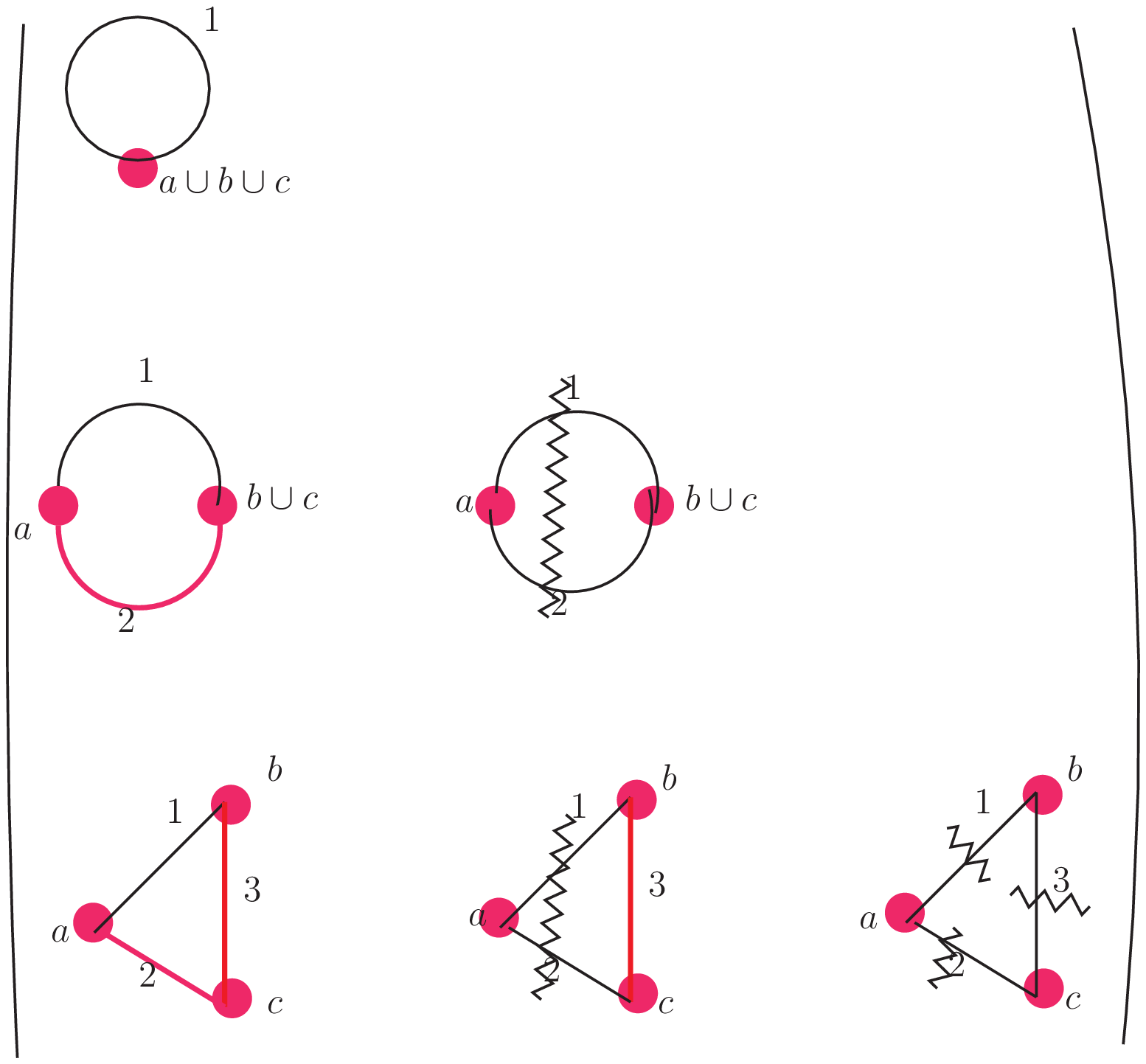}}\;}}
\def\trianglecubical{{\;\raisebox{-40mm}{\epsfxsize=80mm\epsfbox{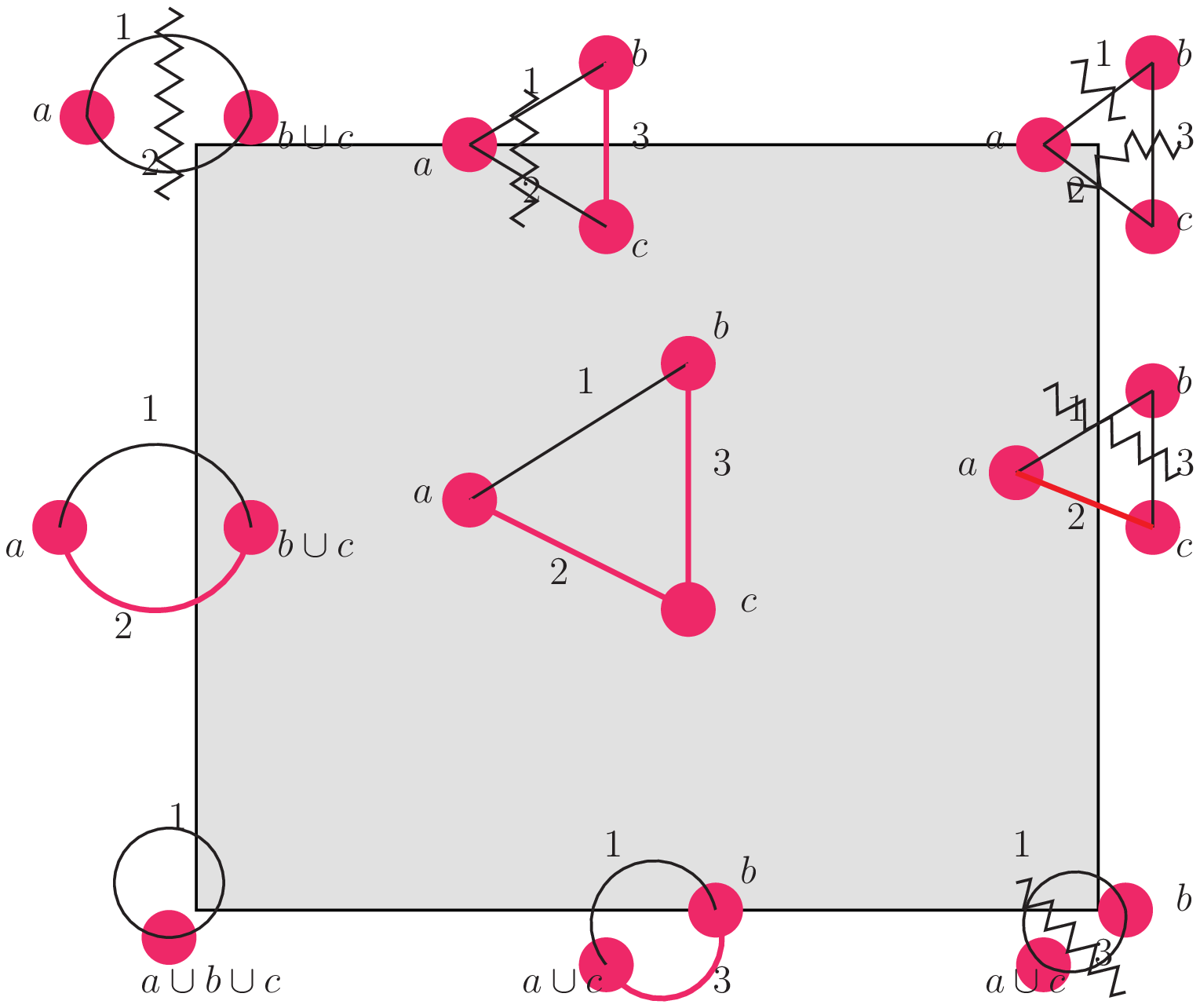}}\;}}
\def\morse{{\;\raisebox{-40mm}{\epsfysize=100mm\epsfbox{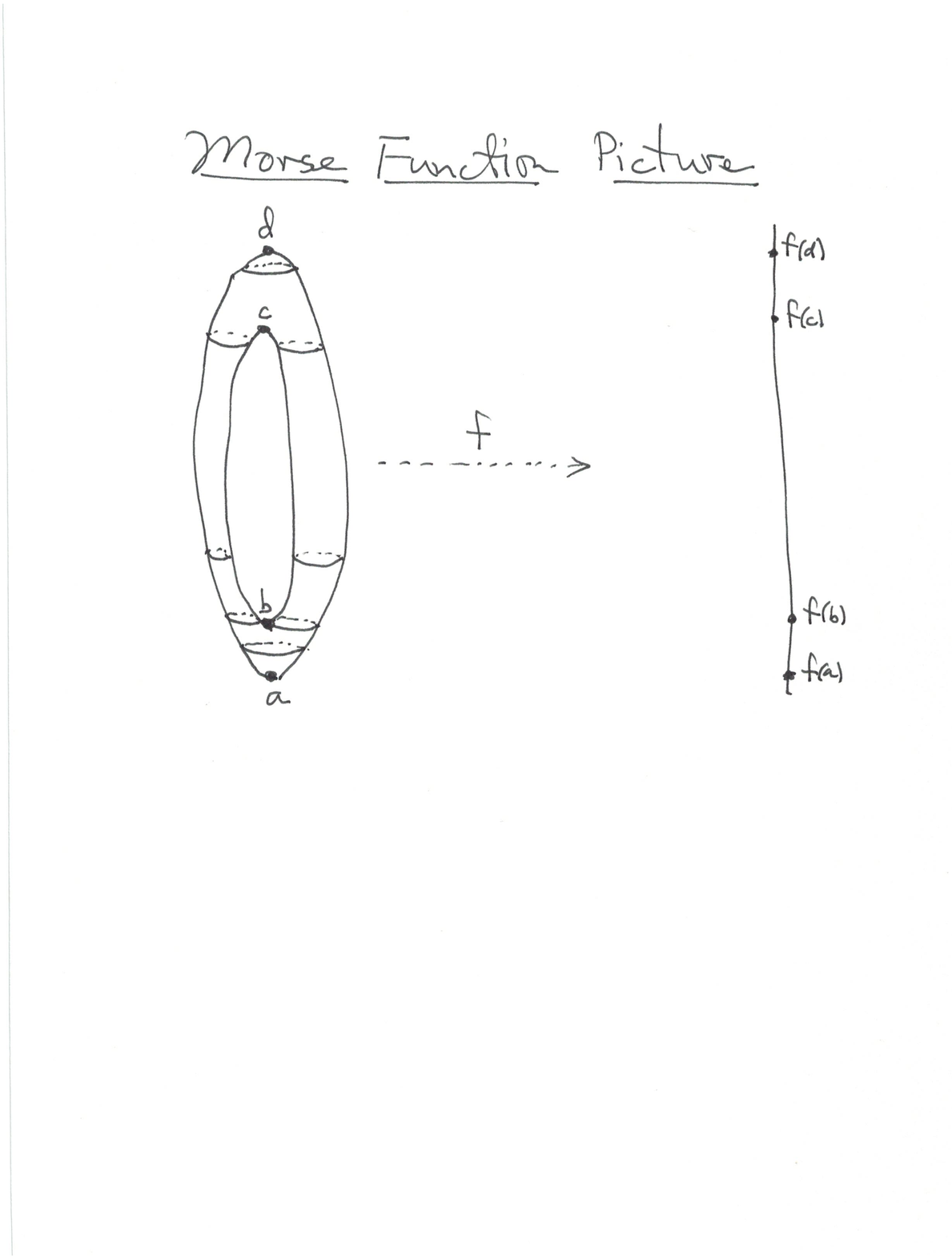}}\;}}

%\vspace*{1.2in}
\title{Cutkosky rules and  Outer Space}
\author{Spencer Bloch \and Dirk Kreimer}
\thanks{DK thanks the Alexander von Humboldt Foundation and the BMBF
for support by an Alexander von Humboldt Professorship.
}\address{Math.\ Dept., Chicago U.\ \and Depts.\ of Math.\ and of Physics, Humboldt U.\ Berlin\\
E-mail address: spencer\_bloch@yahoo.com \and krei\-mer@physik.hu-berlin.de}

\maketitle
\begin{abstract}
We derive Cutkosky's theorem starting from Pham's classical work. 
We emphasize structural relations to Outer Space.
\end{abstract}
\section{Introduction and results}
\subsection{Introduction}
Over the last $50$ years, a technique developed by Landau, Cutkosky, and others has been used to study the variation of Feynman amplitudes as external momenta wind around threshold divisors (discriminants). One chooses a subset of edges of the graph and particular values of the external momenta such that the intersection of the subset of propagator quadrics has an isolated ordinary double point. The presence of such a singularity means that a threshold divisor passes through the given point in external momentum space. Cutkosky rules say the variation of the Feynman amplitude as external momenta wind around the threshold divisor is given (upto a power of $2\pi i$) by putting the subset propagators on shell and integrating over the locus where the energies are $\ge 0$. Techniques developed by Pham \cite{Pham1} can be used to give a rigorous mathematical justification for Cutkosky rules, but curiously, to our knowledge this has never been done. There is a subtlety which is that the mathematical theory of vanishing cycles is usually treated geometrically as a theory describing monodromy for families of complex algebraic varieties. 

For physicists, however, existence of vanishing cycles and the Picard-Lefschetz formula is not sufficient. One needs to know that the vanishing cycle (or more precisely the vanishing sphere in the sense of Pham) is the component of the real locus of the intersection of the quadrics being cut determined by the on-shell condition. This reality condition can be reinterpreted as saying that a Hessian matrix at the singular point is (either positive or negative) definite.

We will embed the study of Cutkosky rules in a study of graphs and their cubical chain complex.
The motivation comes from the fact that understanding of the analytic structure of the contribution of a graph to a Feynman amplitude
is related to an analysis of
its reduced graphs and the graphs in which internal edges are on the mass-shell. The former case relates to graphs in which  internal edges shrink. The latter case relates to graphs with cut edges. The set of cut edges is uniquely determined by the choice of a spanning forest for the graph.

Such pairs of graphs and their spanning forests  populate the cubical chain complex.
The latter is based on a given bridge-free graph $\Gamma$ and a chosen spanning tree $T$ for it.

A given ordering of the edges of $T$ defines then a sequence of spanning forests $F$, and to any pair $(\Gamma,F)$ for fixed $\Gamma$ 
we can associate:\\ 
-a reduced graph $\Gamma_F$ obtained by shrinking all edges of $\Gamma$ to length zero which do not connect different components of the spanning forest,\\ 
-a cut graph $\Gamma^F$ where all those edges connecting different components are put on-shell, so are marked by a Cutkosky cut,\\
-the set of graphs $G^F=\Gamma-E_{\Gamma_F}$ obtained from $\Gamma$ by removing the edges which connect distinct components of the spanning tree. 

Such data define a cell-complex and also a set of lower triangular matrices which allow to analyse a graph amplitude from its reduced graphs and the variations obtained by putting internal edges on-shell.   

The combined data stored in all such matrices contains sufficient information to reconstruct the graph from its variations then through a sequence of dispersion integrals based on the variations as defined by the Cutkosky rules. 

This is made possible as the above sequence of forests allows an iterative analysis of anomalous thresholds, based on the real analyticity of Feynman rules. This is an iterated application of the optical theorem and dispersion using real analyticity with respect to a single kinematical variable defined for each $F$.
We will exhibit this in a final section where we discuss an example from physics.

\subsection{Results}
Our first result is Cutkosky's theorem. It suffices to derive it in the case that the disjoint components of $G^F$ do not contain loops.

We also provide an edge-by-edge analysis of a graph using properties of the parametric representation for Feynman diagrams.

When we combine Cutkosky's theorem with straightforward properties of graph polynomials, we arrive at a method 
to analyse monodromies related to amplitudes multiple-edge by multiple-edge.

A sequence of cuts 
\[\epsilon_2\to \epsilon_3\to\cdots\to \epsilon_{v_\Gamma}\] 
will shift the normal threshold $s_0(\epsilon_2)$ associated with a chosen cut $\epsilon_2$ to anomalous thresholds
\[
s_0(\epsilon_2)\to s_1(\epsilon_3)\to\cdots\to s_{v_\Gamma-2}(\epsilon_{v_\Gamma}).
\]
\begin{remark}
The resulting sequence of anomalous thresholds $s_i(\epsilon_{i+2})$, $i>0$ is a sequence of values for a channel variable $s$ defined by $\epsilon_2$.
They are computed from the divisors associated to $\epsilon_{i+2}$. The latter are functions of all kinematical variables.
For example, for the one-loop triangle discussed in the final section the divisor in $\mathbb{C}^3$ associated to $\epsilon_3$ is a simple function of 
\[
\lambda(p_1^2,p_2^2,p_3^2)=p_1.p_2^2-p_1^2p_2^2=p_2.p_3^2-p_2^2p_3^2=p_3.p_1^2-p_3^2p_1^2,\,p_1+p_2+p_3=0.
\] 
The three representations of $\lambda$ allow to compute $s_1(\epsilon_3)$ for $s=p_3^2$ or $s=p_1^2$ or $s=p_2^2$ respectively.
\end{remark}

As a result, to a graph $\Gamma$  we can assign a collection of lower triangular matrices $M_i^\Gamma$ with the following properties:
\begin{itemize}
\item All entries in the matrix correspond to well-defined integrable forms under on-shell renormalization conditions.
\item Anomalous thresholds $s_i$ are determined from properties of graph polynomials. They provide lower boundaries for dispersion integrals associated to these integrable forms. 
\item Along the diagonal in the matrices $M_i^\Gamma$ we find leading threshold entries: all quadrics for all edges in a graph  are on the mass-shell.
\item The variation of a column in $M_i^\Gamma$ wrt to a given channel is given by the column to the right. 
\item Non-leading thresholds (entries below the diagonal ) correspond to fibrations over the leading thresholds. They are cones if there are loops in the uncut edges.
\item The subdiagonal entries $(M_i^\Gamma)_{k,k-1}$ are determined from the diagonal entries  $(M_i^\Gamma)_{k-1,k-1}$
and $(M_i^\Gamma)_{k,k}$ via a dispersion integral. This gives $(k-1)$ two-by-two matrices each of which has an interpretation via the optical theorem.
This hence determines the first subdiagonal.
\item Continuing, all subdiagonals and hence the whole matrix $(M_i^\Gamma)_{r,s}$ is determined via iterated dispersion.
This answers the question how to continue the optical theorem beyond two-point functions.
\item The first column in each matrix $M_i^\Gamma$ corresponds to a path in the spine of Outer Space from some rose to some cell containing $\Gamma$.
The entries of  all such matrices which one can assign to a pair $(\Gamma,T)$ by different choices of the ordering of the edges of $T$ give a cell for $(\Gamma,T)$ in the cubical chain complex.
\item There are as many distinct matrices $M_i^\Gamma$ as there are distinct paths from roses to $\Gamma$.
\item Graphs in Outer Space are metric graphs. The Feynman integral in parametric space corresponds to an integral over a cell in Outer Space.
\item The markings at a graph in Outer Space fix the ambiguities related to the variations associated with thresholds. In particular this last aspect we reserve to future investigation.
\end{itemize}
In fact, this first paper on the subject serves only to settle ideas and provide a starting point for future investigations.
In particular, we will treat the matrices $M^\Gamma_i$ combinatorially, and take them to have as entries the modified graphs rather than the asssociated amplitudes.
Only when we turn to physics examples we will consider their entries as given by the associated amplitudes.

The relation  to outer space and a Hodge theoretic analysis of the matrices $M_i^\Gamma$ certainly call for more detailed analysis in future work.

Here, we give a mathematical result -precise formulation of Cutkosky's thorem- and analyse normal and anomalous thresholds of amplitudes -physics notions- in terms of parametric representation of Feynman integrals. Outer space and its cubical chain complex provides a common structure to this analysis which fascinates us both.
\subsection*{Acknowledgments.} Both authors thank Karen Vogtmann and Holger Reich for helpful discussions.

\section{Graphs and spanning forests}
We let $\Gamma$ be a connected graph. We allow multiple edges between vertices and self-loops as well.

We let $V_\Gamma$ then be the set of vertices of $\Gamma$, $|V_\Gamma|=v_\Gamma$, and $E_\Gamma$, $|E_\Gamma|=e_\Gamma$ be the set and number
of edges.

$\Gamma/X$, for $X\subseteq\Gamma$ a (not necessarily connected) graph, denotes the graph obtained from $\Gamma$ by shrinking all internal edges of $X\subseteq\Gamma$ to zero length:
\be \Gamma/X=\Gamma_{|l(e)=0,e\in E_X}.\ee

A spanning tree $T$ of $\Gamma$ is a proper subgraph $T\subseteq\Gamma$ such that $V_T=V_\Gamma$ and $T$ is connected and simply connected. 

A spanning $k$-forest is a disjoint union $\amalg_{i=1}^k T_i$ of $k$  trees $T_i\subsetneq \Gamma$,
and $\cup_i V_{T_i}=V_\Gamma$.

A pair $(\Gamma,F)$, $F$ a spanning $k$-forest for $\Gamma$, defines a set $G^F$ of $k$ mutually disjoint  graphs $\Gamma_j\subsetneq \Gamma$,
 $1\leq j\leq k$. For each edge $e\in \Gamma_i$  in such a graph $\Gamma_i$ we have that its boundary vertices $\partial(e)=\{v_+(e),v_-(e)\}\in V_{\Gamma_i}$ belong to the vertices of  the same graph.

$F$ also defines a unique set of edges $E_{\Gamma_F}$ which connect vertices of different such $\Gamma_i$ and such that $\Gamma_F:=\Gamma/(\cup_{i=1}^k \Gamma_i)$ is based on those
edges and $k$ vertices. In particular, for a 2-forest, we obtain two vertices connected by a multiple edge.

$\Gamma-X$, for $X\subseteq \Gamma$, denotes the graph obtained from $\Gamma$ by removing the edges of $X$. It can contain isolated vertices.

A graph is bridge-free or 1PI or 2-connected if $(\Gamma-e)$  is connected for any $e\in E_\Gamma$.

A connected graph $\Gamma$ which is not bridge-free is a union of 1PI graphs $\gamma=\cup_i\gamma_i$ and trees $T=\cup_i T_i$ such that $\Gamma/\gamma=T$
and $\Gamma/(\cup_i T_i)=\gamma$.

Here we associate the disjoint union of graphs $\gamma=\cup_i\gamma_i$ with the connected graph obtained by shrinking
all bridges (edges of the trees $T_i$), and similarly for $T$.

A graph $\Gamma$ is 1VI (one-vertex irreducible) if $\Gamma-v$ is connected for any vertex $v\in V_\Gamma$. Trees of edges are always one-vertex reducible.
%A graph which is not 1VI is identified with a suitable union of 1VI graphs as in \cite{BlKr}. 

We let $|\Gamma|:=|H_1(\Gamma)|$ be the first Betti number. Note $|\Gamma|=|\Gamma/F|$ for any spanning forest $F$ of $\Gamma$.

For disjoint unions of graphs $h_1,h_2$, we set $|h_1\cup h_2|=|h_1|+|h_2|$.

The length of a path is the number of edges in the path. 
The distance between two vertices is the number of edges in the shortest path between them, a pair of vertices of unit distance is called adjacent.

\subsection{Momenta and masses}
To each vertex $v\in V_\Gamma$, we assign a four-momentum $p_v\in \mathbb{M}^4$. For the purposes of this section, we can let components of these vectors be real. 

We require momentum conservation:
\be \sum_{v\in V_\Gamma}p_v=0.\ee

By definition, if vertices $v_i \in V_\Gamma$ merge together to a vertex $w=\cup_i v_i$  in $\Gamma/X$, the vertex $w$ has momentum $\sum_i p_{v_i}$ assigned.

Also, to all internal edges $e\in E_\Gamma$, we assign a parameter $m_e^2-i\epsilon$, $m_e,\epsilon>0$ real and $\epsilon \ll 1$ infinitesimal
(often written $m_e^2-i0$ in the physics literature).

\subsection{Powercounting}
To all edges $e$ and vertices $v$ of a graph, we assign weights $w(e),w(v)\in\mathbb{Z}$, and define the weight of a graph to be the integer
\be
w(\Gamma)=4|\Gamma|-\sum_{v\in V_\Gamma} w(v)-\sum_{e\in E_\Gamma}w(e).
\ee  
We always set $w(v)=0$ and $w(e)=2$ below.

\subsection{Multiple edges}
If the set of edges connecting two adjacent vertices $x,y$ has cardinality $k$, we call this set a multiple edge $b_k(x,y)$, a $k$-edge banana (often denoted by $b_k$ if the vertices are clear). $b_1$ is a single edge.

For two adjacent vertices $x,y$ in a bridge-free graph $\Gamma$
connected by a $k$-edge banana, $\Gamma-b_k(x,y)$ is regarded as a set of trees $T_i$ and bridge-free components $\gamma_i$ as above. 

Edges and multiple edges will play an important role, as we will analyse graphs by removing them, or shrinking them.
\subsection{Example}
The following figure gives an example.
$$
\Gg
$$
In the first row on the left we see the graph $\Gamma$ on edges $1,\ldots,8$. We let $\gamma\subset\Gamma$ be the subgraph on edges $3,4$.
In the first row next to $\Gamma$ we see $\Gamma-\gamma$, and below in the second row to the left  we have $\Gamma/\gamma$. 
On the right in the first row, we see the tree $T$ on edges $7,8$ obtained by shrinking the two subgraphs on edges $1,2$ and edges $5,6$ to points. On the other hand, shrinking edges $7,8$ in $\Gamma-\gamma$, we get the two subgraphs 
as a new graph $g$ which is not 1VI: the two vertices of edge $7$ unify to the new vertex which connects the two subgraphs to the graph $g$
in the second row on the right.
\section{The cubical chain complex and $M_\Gamma^i$}
\subsection{The cubical chain complex}
We follow  \cite{VogtmannHatcher}.

Consider a pair $(\Gamma,T)$ of a bridge free graph $\Gamma$ and a chosen spanning tree $T$ for it.
Assume $T$  has $k$ edges. Consider the $k$-dimensional unit cube. It has origin $(0,\cdots,0)$ and $k$ unit vectors
$(1,0,\cdots,0),\ldots$, $(0,\cdots,0,1)$ form its edges regarded as $1$-cells. A change of ordering of the edges of $T$ permutes those edges.

The origin is decorated by a rose on $|\Gamma|$ petals, and the corner $(1,1,\cdots,1)$ decorated by $(\Gamma,V_\Gamma)$, with $k=v_\Gamma-1$,
and we regard $V_\Gamma$ as a spanning forest.

The complex is best explained by assigning graphs as in the following example.
$$\cubical$$
The cell is two-dimensional as each of the five spanning trees of the graph $\Gamma$, the dunce's cap graph, in the middle of the cell has length two.

We have chosen a spanning tree $T$ provided by the edges $e_1$ and $e_3$, indicated in red. The boundary of our two-dimensional cell has four one-dimensional edges, bounded by two of the four 0-dimensional corners each.

To these lower dimensional cells we assign graphs as well.

To the vertical edge on the left we assign the pair $(\Gamma/e_1,e_3)$, in which the first edge, $e_1$, of the spanning tree $T$ shrinks to zero length. We have $\Gamma/e_1=\Gamma_{T-e_1}$, as $F=T-e_1$.
The edge $e_3=T/e_1$ remains as its spanning tree. The graph has two vertices, $a\cup c$ is the vertex onto which vertices $a,c$ collapse, the other vertex is $b$.

To the vertical edge on the right we assign the pair $(\Gamma, a\cup e_3)$, $a\cup e_3\equiv T-e_1$, with the union of the vertex $a$ with edge $e_3$ a spanning forest for $\Gamma$,
obtained by removing edge $e_1$ from $T$. $\Gamma^{T-e_1}$ is the indicated graph with the cut edges marked by a zigzag line.

To the horizontal edge below we assign the pair $(\Gamma/e_3,e_1)$,in which the second edge, $e_3$, of the spanning tree $T$ shrinks to zero length.
The edge $e_1$ remains as its spanning tree. The graph has two vertices, $b\cup c$ is the vertex onto which vertices $b,c$ collapse, the other vertex is $a$.

To the horizontal edge above  we assign the pair $(\Gamma, b\cup e_1)$, with the union of the vertex $b$ with edge $e_1$ a spanning forest for $\Gamma$,
obtained by removing edge $e_3$ from $T$.

The lower left corner carries the graph obtained by shrinking both edges of the spanning tree, a rose on two petals with a single vertex $a\cup b\cup c$.

The lower right corner carries the graph obtained by shrinking edge $e_3$ and removing edge $e_1$ from the spanning tree.

The upper left corner carries the graph obtained by shrinking edge $e_1$ and removing edge $e_3$ from the spanning tree.

The upper right corner finally carries the graph obtained by removing both edges from the spanning tree, so that the forest is the disjoint union of the three vertices.

The spanning tree has length two and so there are $2=2!$ orderings of its edges, and hence two lower triangular $3\times 3$ matrices  $M^\Gamma_i$ which we can assign to this cell.

They look as follows:
$$\cubicalM $$
These square matrices are lower triangular. An entry $(M^\Gamma_i)_{rs}$, $1\leq r\leq v_\gamma$, $1\leq s \leq v_\gamma$ is given
by 
\be 
(M^\Gamma_i)_{rs}=(\Gamma/f_{v_\Gamma-r},T/f_{v_\Gamma-r}-g_{s-1}).
\ee
Here, $\Gamma/f_i$ is the graph obtained by shrinking edges $e_1,\cdots e_i$ of the spanning tree. For $\Gamma/f_0=\Gamma$ we shrink none.
$g_0$ is the empty set, and $g_r$ is the union of the last $r$ edges of the spanning tree. 

Let us now look at these matrices from the viewpoint of a partition of the set of vertices.
\subsection{The matrix $M_i^\Gamma$}\label{matrixM}
Let $\Gamma$ be a connected graph with vertex set $V=V_\Gamma$ and edge set $E=E_\Gamma$. We will be interested in partitions of $V$
\eq{}{\sP = V_1\amalg\cdots \amalg V_r.
}
Associated to a partition is a collection of subgraphs $G^F=\{\Gamma_i\}$, where the edges of $\Gamma_i$ are all edges $e$ of $\Gamma$ such that the vertices $\partial e = \{v^+_e, v^-_e\}$ both lie in $V_i$ as before. 

We say that the partition $\sP$ is {\it connected} if all the $\Gamma_i$ are connected graphs. Note that if $\Gamma_i$ is not connected, there is a natural refinement of $\sP$ given by further partitioning $V_i$ according to the connected components of $\Gamma_i$. Unless otherwise specified, we will work throughout with connected partitions. 

Given a connected partition $\sP(V(\Gamma))$, define
\ml{}{E_{\Gamma_F}\equiv E_\sP := E(\Gamma)\backslash(\coprod_i E(\Gamma_i))= \\
\{e\in E(\Gamma)\ |\ \text{vertices of $e$ lie in distinct $V_i$}\}.
}

Let $\Gamma//\coprod \Gamma_i$ be the graph obtained from $\Gamma$ by shrinking all the $\Gamma_i$ to (distinct) points. (Note that the $\Gamma_i$ are disjoint.) Our assumption that the $\Gamma_i$ are connected is easily seen to imply an exact sequence on the level of first homology with $\Z$-coefficients:
\eq{}{0 \to \bigoplus_i H_1(\Gamma_i) \to H_1(\Gamma) \to H_1(\Gamma//\coprod_i \Gamma_i) \to 0. 
}
We write
\eq{}{\Gamma_F\equiv \Gamma_\sP := \Gamma//\coprod_i \Gamma_i. 
}
We have $E_{\Gamma_\sP} = E_\sP$. 

Given a connected partition $\sP = \sP(\Gamma)$ the corresponding {\it cut} graph is by definition the graph $\Gamma$ with the edges $E_\sP$ removed (``cut'')\footnote{For the purposes of this section, edges which are on-shell, i.e.\ 'cut', are simply removed.}. I.e. the cut graph is $\coprod_i \Gamma_i$. Note that some of the $\Gamma_i$ may be isolated vertices. In other words, cutting an edge does not mean removing its endpoints. We also use the notation $(\Gamma,\sP)$ a bit abusively to denote the pair $(\Gamma, E_\sP)$. Thus $(\Gamma,\sP)$ represents the graph $\Gamma$ together with the choice of a cut $E_\sP$. 

We are particularly interested in the case of connected $2$-partitions $V_\Gamma=V_1\amalg V_2$. 
\begin{lem}\label{connected} Let $\Gamma$ be a connected graph with at least $1$ edge and no self-loops. Then $\Gamma$ admits a connected $2$-partition.
\end{lem}
\begin{proof}For $v\in V_\Gamma$, write $\Gamma_v$ for the graph obtained from $\Gamma$ by removing $v$ and all edges containing $v$. We claim there exists a $v$ such that $\Gamma_v$ is connected. This is clear if $\Gamma$ is a tree, because we can take $v$ to be a unary vertex. If $\Gamma$ is not a tree, then there exists an edge $e$ such that $\Gamma-e$ is still connected. Since $\Gamma$ is not a tree, $\Gamma-e$ has at least one edge and no self-loops, so by induction on the number of edges, there exists $v\in \Gamma-e$ such that $(\Gamma-e)_v$ is connected. But this implies $\Gamma_v$ is connected as well.  The $2$-partition $\{v\}\amalg V_{\Gamma_v}$ is thus connected. 
\end{proof}

Partitions can be refined in an evident sense. A refinement of the partition $V_\Gamma = \coprod_i V_i$ consists of refinements $V_i = \coprod V_{i,j}$ for all $i$. (We admit the trivial refinement $V_i=V_i$.) The refinement is connected if all the $\Gamma_{i,j}$ are connected. We refer to a sequence of refinements as {\it binary} if each non-trivial refinement is a connected $2$-partition, viz. $V_{i_1,\dotsc,i_k} = V_{i_1,\dotsc,i_k,1}\amalg V_{i_1,\dotsc,i_k,2}$.  

Recall a {\it spanning tree} in a connected graph $\Gamma$ is a subgraph $T\subset \Gamma$ which is a tree (connected, with no loops) containing all vertices of $\Gamma$. Suppose we are given a connected $2$-partition $\sP: V_\Gamma=V_1\amalg V_2$. Let $T_i$ be a spanning tree for $\Gamma_i, i=1,2$, and let $e\in E_\Gamma$ be an edge in $\Gamma_\sP$. Then 
\eq{}{T = T_1\cup \{e\}\cup T_2
}
is a spanning tree for $\Gamma$. Such a $T$ is said to be adapted to the partition $\sP$. More generally, a spanning tree $T$ is adapted to a connected binary refinement if it contains exactly one edge from each ``subcut''. For example, if $V=(V_{11}\amalg V_{12})\amalg V_2$ is a connected refinement of $V=V_1\amalg V_2$, then an adapted spanning tree would have the form
\eq{}{T = T_{11}\cup T_{12}\cup e_{1;1,2} \cup e_{1,2}\cup T_{2}.
}
Here $e_{1;1,2}$ is an edge of $\Gamma_1$ connecting the vertex sets $V_{11}$ and $V_{12}$. 

There is a maximal (finest) partition of $V_\Gamma$ where each $V_i$ consists of a single vertex. As a consequence of lemma \ref{connected} we can find a sequence of connected binary refinements of a given connected partition abutting to the maximal partition. Fix such a maximal sequence $\sP$, and let $T$ be a spanning tree for $\Gamma$ which is adapted to it. We want to associate to $(\Gamma, \sP, T)$ a matrix $M$ whose entries are ``graphs with cuts''. 

We define a height function $ht: E_T \to \{1,2,3,\ldots\}$ by taking $ht(e) = k$ if $e\in E_T$ connects $\Gamma_{i_1,\dotsc,i_{k-1},1}$ to $\Gamma_{i_1,\dotsc,i_{k-1},2}$. For example $ht(e)=1$ means $e$ connects $V_1$ and $V_2$. Note that for $k>1$, a given height $k$ may correspond to several cuts. We choose some ordering for these cuts compatible with the heights and display the graphs from left to right:
\ml{}{\Gamma,\ \ \Gamma(\text{unique ht $1$ cut}), \ \ \Gamma(\text{ht $1$ cut, one ht $2$ cut}), \\
\Gamma(\text{ht $1$ cut, two ht $2$ cuts})\cdots \Gamma(\text{all cuts}).
}
Next, over each entry $\Gamma(\text{...cuts...})$ we move upward shrinking edges from $T$ starting with the edges of largest height. Again there is some ambiguity but we choose an ordering among edges of equal height. By shrinking edges of largest height while cutting edges of smallest height, we build a triangular matrix of graphs with entries along the diagonal having only cut edges and tadpoles. 
Matrices $M^\Gamma_1,M^\Gamma_2$ above are examples.
%$$\graphcutshrink$$

\subsection{Graphs and their Hopf algebra structure}

\subsubsection{Hopf algebra structure}
Consider the free commutative $\mathbb{Q}$-algebra 
\be 
H=\oplus_{i\geq 0} H^{(i)},\, H^{(0)}\sim \mathbb{Q}\One,
\ee
generated by 2-connected graphs as free generators (disjoint union is product $m$, labelling of edges and of vertices by momenta as declared).

Consider the Hopf algebras $H(m,\One,\Delta,\hat{\One},S)$ and $H(m,\One,\Delta_c,\hat{\One},S_c)$, given by
\be \One:\mathbb{Q}\to H, q\to q\One,\ee
\be \Delta:H\to H\otimes H, \Delta(\Gamma)=\Gamma\otimes\One+\One\otimes\Gamma+\sum_{\gamma\subsetneq\Gamma,\gamma=\cup_i\gamma_i,w(\gamma_i)\geq 0}
\gamma\otimes \Gamma/\gamma,\ee
\be \Delta_c:H\to H\otimes H, \Delta_c(\Gamma)=\Gamma\otimes\One+\One\otimes\Gamma+\sum_{\gamma\subsetneq\Gamma,\gamma=\cup_i\gamma_i}
\gamma\otimes \Gamma/\gamma,\ee
\be \hat{\One}:H\to\mathbb{Q}, q\One\to q, H_>\to 0,\ee
\be S: H\to H, S(\Gamma)=-\Gamma-\sum_{\gamma\subsetneq\Gamma,\gamma=\cup_i\gamma_i,w(\gamma_i)\geq 0}
S(\gamma) \Gamma/\gamma,\ee
\be S_c: H\to H, S(\Gamma)=-\Gamma-\sum_{\gamma\subsetneq\Gamma,\gamma=\cup_i\gamma_i}
S(\gamma) \Gamma/\gamma,\ee
where $H_>=\oplus_{i\geq 1} H^{(i)}$ is the augmentation ideal.

Both Hopf algebras will be needed in the following for renormalization in the presence of variations.

Edges in $\Gamma$ which have one vertex in $\Gamma-\gamma$ and the other in $\gamma$ merge in $\Gamma/\gamma$ into $k$  distinct vertices $v_i$.
In this merging, momentum labels at vertices are additive as stated before:
$q_{v_i}=\sum_{v\in V_{\gamma_i}}q_v$. 

Here is an example for both (reduced) coproducts.
$$
\ExCop
$$
The first reduced coproduct $\tilde{\Delta}$ only produces a single term: the subgraph $\gamma$ on edges $3,4$ has $w(\gamma)=0$ and hence contributes,
while all other subgraphs have a weight $<0$.

We have no such restriction on the other reduced coproduct and hence collect the indicated three terms on the right.

\subsubsection{The Hopf algebra and pairs $(\Gamma,F)$}
Let $(\Gamma,T)$ be a pair of a graph and a spanning tree for it with a choice of ordering for its edges.
Let $\mathcal{F}_{(\Gamma,T)}$  be the set of corresponding forests.

Then, to any pair $(\Gamma,F)$, with $F$ a $k$-forest ($1\leq k\leq v_\Gamma$), $F\in \mathcal{F}_{(\Gamma,T)}$ we can assign a set of $k$ disjoint graphs $G^F$. We let $\Gamma_F:=\Gamma/G^F$ be the graph obtained by shrinking all internal edges of these graphs.

For each such $F$, we call $E_{\Gamma_F}$ a cut. In particular, for $F$ the unique 2-forest assigned to $T$ (by removing the first edge from the ordered edges of $T$), we call $\epsilon_2=E_{\Gamma_F}$ the Cutkosky cut of $(\Gamma,T)$. 

Note that the ordering of edges defines an ordering of cuts
$\emptyset=\epsilon_1\subsetneq\epsilon_2\subsetneq\cdots\subsetneq\epsilon_k=E_\Gamma$.   

For a Cutkosky cut, we have $G^F=(\Gamma_1,\Gamma_2)$ and we call 
\be
s=(\sum_{v\in V_{\Gamma_1}}q_v)^2=(\sum_{v\in V_{\Gamma_2}}q_v)^2
\ee
the channel associated to $(\Gamma,T)$.

These notions are recursive in an obvious way: the difference between a $k$ and a $k+1$ forest defines a Cutkosky cut for some subgraph.

We define $|G^F|=\sum_{\gamma\in G^F}|\gamma|$. Also, we let $\mathcal{F}_k(\Gamma)$ be the set of all $k$-forests for a graph $\Gamma$.

For a disjoint union of $r$  graphs $\gamma=\cup_{i=1}^r\gamma_i$, we say a disjoint union of trees $T=\cup_i t_i$ spans $\gamma$
and write $T|\gamma$, if $t_i$ is a spanning tree for $\gamma_i$. 

We have then an obvious decomposition of all possible spanning forests using the coproduct $\Delta_c$. A spanning forest decomposes into a spanning forest which leaves no loop intact in the cograph together with spanning trees for the subgraph: 
\begin{lem}\label{cccuts}
\be 
\sum_{T|\Gamma^\prime}\left(\Gamma,T\cup \sum_{k=1}^{v_\Gamma}\sum_{F\in\mathcal{F}_k(\Gamma{\prime\prime}),|G_F|=0}F\right)=\sum_{k=1}^{v_\Gamma}\sum_{F\in\mathcal{F}_k(\Gamma)}(\Gamma,F).
\ee
\end{lem}
\begin{remark}
This lemma ensures that uncut subgraphs with loops can have their loops integrated out. The resulting integrals are part of the integrand of the full graph and its variations determined by the cut edges. 
\end{remark}
\section{Kinematics}
For a graph $\Gamma$ with momenta $q_v$ and masses $m_e$, we consider 
the real vectorspace 
\be \label{qgamma}
Q_\Gamma=\mathbb{R}^{e_\Gamma+v_\Gamma(v_\Gamma-1)/2}\
\ee
 of dimension $e_\Gamma+v_\Gamma(v_\Gamma-1)/2$
generated by $m_e^2$ and $q_v\cdot q_w$ (ignoring constraints from finite dimensionality of Minkowski space). 

Note that $Q_\Gamma/(Q_{\Gamma/b_k})$ is a real vector space of dimension $v_\Gamma-1+k$.

A point $p\in Q_\Gamma$ is of the form
\be 
p=\sum_{e\in E_\Gamma} r_e m_e^2+\sum_{i\leq j, i,j\in V_\Gamma}
r_{ij}q_i\cdot q_j,
\ee
with $r_r,r_{ij}\in\mathbb{R}$ and $q_i\cdot q_j:=\sum_{\mu=0}^3 q_{i\mu}q_j^\mu=q_i^0q_j^0-q_i^1q_j^1-q_i^2q_j^2-q_i^3q_j^3$.

The quadric $l(e)\equiv Q(e)$ assigned to edge $e$ with momentum $q(e)$ and mass $m(e)$
is 
\be 
Q(e):=q(e)\cdot q(e)-m_e^2.
\ee 
We say an edge $e$ is on shell if $l(e)=0$ and $q^0(e)>0$.

Here, $q(e)$ is fixed by assigning an internal momentum $k(e)$ to each oriented edge $e$ outside a spanning tree.
These edges together with a choice of a spanning tree define a basis for $H^1(\Gamma)$, and we require that $k(e)$, a loop momentum, is routed through the spanning tree in the orientation given by $e$.
Furthermore we require that external momenta assigned to vertices traverse only through edges in the spanning tree and we require  momentum conservation at each vertex. All external momenta are considered incoming.

\section{Landau singularities}
We consider the universal quadric
\eq{1}{Q: \vartheta=\sum_{i=1}^{e_\Gamma} a_i \ell_i= 0;\quad \ell_i= q_i^2-m_i^2
}
Here $Q\inj \P^{e_\Gamma-1} \times \C^{gD}$ ($D=4$ in four dimensions of spacetime). The external momenta and masses are viewed as fixed. Writing $x_{ij}, 1\le i\le g, 1\le j\le D$ for the coordinate functions on $\C^{Dg}$, we have $\vartheta = \vartheta(a,x)$ and we are interested in singular points of $\vartheta$. 

We can complete the square and rewrite (\cite{SMatrix})
\eq{2}{\vartheta(a,x) = \vec{x'}A\vec{x'}^t+\Phi(a)/\psi(a)
}
Here $\psi(a)=\det(A)$ is the first Symanzik, $x'$ denotes a certain affine linear change of coordinates, and $\phi(a)$ is the second Symanzik. (We omit masses and external momenta from the notation because we view these as fixed.) 

We consider a point $(a^0,x^0)$ which is a singular point of $Q$ such that $\psi(a^0)\neq 0$. For a physical singularity, all the $a^0_e>0$ so this condition holds. ($\psi$ is a sum of monomials with coefficient $+1$.) We can, in fact, weaken this positivity condition. We might say a singularity $(a^0,x^0)$ of $Q$ is {\it weakly physical} if all $a_e^0 \ge 0$ and the set of edges $e$ such that $a_e^0=0$ does not support a loop on the graph. It is straightforward to check in this case that $\psi(a^0) > 0$.

\begin{lem}The mapping $(a^0,x^0) \mapsto a^0$ gives a bijection between weakly physical singularities of $Q$ and singularities $a^0$ of the second Symanzik $\phi(a)$ such that $a^0_e\ge 0$ for all edges $e$ and $\psi(a^0)\neq 0$. 
\end{lem}
\begin{proof} One can obtains the locus $\Phi(a)/\psi(a)=0$ in $\P^{e_\Gamma-1}-\{\psi=0\}$ by intersecting $Q$ with the $gD$ affine linear forms $\partial\vartheta/\partial x_{ij}=0$. Said another way, this intersection in $(\P^{e_\Gamma-1}-\{\psi(a)=0\})\times \C^{gD}$ projects isomorphically to the locus $\phi(a)/\psi(a)=0$ in $\P^{e_\Gamma-1}-\{\psi=0\}$. Since $\partial\vartheta/\partial x_{ij}$ necessarily vanishes at any singular point of $Q$, and since each intersection cuts down the tangent space dimension by at most one, we see that singular points on $Q$ give rise to singular points on the intersection. 

Conversely, given a singular point $a^0$ of $\phi$ such that $\psi(a^0)\neq 0$, it follows from \eqref{2} that the point $(a^0, x''=0)$ is a singular point of $Q$. 
\end{proof}

A Landau singularity is a point $p=(p_A,p_Q)\in \mathbb{P}^{e_\Gamma-1}\times Q_\Gamma$ such that 
\be
\Phi(p)=0, \left(\frac{\partial}{\partial A_e}\Phi_e\right)(p)=0,\,\forall e.
\ee

A Landau singularity is leading if $p_A$ lies in the interior of the simplex $\sigma_\Gamma$,
\be \label{simplex}
\sigma_\Gamma:\{\prod_e A_e=0\},
\ee
 so if it is physical as opposed to weakly physical.

%A secondary singularity is a point $p=(p_A,p_Q)\in \mathbb{P}^{e_\Gamma-1}\times Q_\Gamma$ such that 
%\be
%\Phi(p)=0, \Phi_e(p)=0,\,\forall e,
%\ee
%and $p_Q\in Q_\Gamma^0$.
\section{Thresholds}
The cut $\epsilon_2$ defines a normal threshold given by  
\be
s_0=s_0(\epsilon_2)= (\sum_{e\in E_{\Gamma_2}}m_e)^2,
\ee
the corresponding pseudo thresholds allow for negative signs:
\be
 (\sum_{e\in E_{\Gamma_2}}\pm m_e)^2,
\ee
and are unphysical thresholds.

The normal threshold corresponds to a Landau singularity of the reduced graph $\Gamma_2$. Given as a singular point of the second Symanzik polynomial $\Phi_{\Gamma_2}$ 
it is at  $p=(p_A,p_Q)$:
\be p_Q: s=(\sum_{e\in E_{\Gamma_2}}m_e)^2,\ee 
and
\be 
p_A: \{A_im_i=A_jm_j\}.
\ee

We call the thresholds $s_{j-2}$ assigned to cut $\epsilon_j$, $j>2$, anomalous.
They are defined with respect to the channel $s$  defined by the Cutkosky cut $\epsilon_2$.

They can be computed studying discriminants of graph polynomials as points on a one-dimensional real subspace of $Q_\Gamma$ defined by the associated channel $s$.
See Thm.(\ref{anomalous}).

\begin{remark}
The divisor assigned to a cut $\epsilon_j$ is a function of  complex variables in $Q_{\Gamma_{F_j}}\subseteq Q_\Gamma$, $F_j$ the $j$-forest assigned to $\epsilon_j$. Physical Landau singularities determine a function $s(p_A),p_A\in \mathbb{P}^{e_{\Gamma_{F_j}}-1}$ for fixed real kinematical variables on the co-dimenion one subspace of $Q_{\Gamma_{F_j}}$ determined by the channel variable $s$. Anomalous thresholds minimize this function.
\end{remark}

\begin{remark}
If we shrink all edges in a spanning tree $T$ of $\Gamma$ we are left with  $|\Gamma|$ edges which constitute a rose $r_{|\Gamma|}$  with $|\Gamma|$ petals. 
There are no possible cuts,  and hence no threshold.
\end{remark}

\section{Graph Polynomials and Feynman rules}
\subsection{Symanzik polynomials}
Let $\psi(\Gamma),\phi(\Gamma)$ be the two usual graph polynomials, and 
\be \Phi(\Gamma)=\phi(\Gamma)-M(\Gamma)\psi(\Gamma),\ee
the full second graph polynomial with masses.
Here,
\be 
M(\Gamma):=\left(\sum_{e\in E_\Gamma}m_e^2 A_e\right).
\ee
We have 
\be
\psi(\Gamma)=\psi(\Gamma/\gamma)\psi(\gamma)+R_\gamma^\Gamma,
\ee
\be
\phi(\Gamma)=\phi(\Gamma/\gamma)\psi(\gamma)+\tilde{R}_\gamma^\Gamma.
\ee
\be
\Phi(\Gamma)=\phi(\Gamma/\gamma)\psi(\gamma)+\bar{R}_\gamma^\Gamma.
\ee
\be
\psi(\Gamma_1\Gamma_2)=\psi(\Gamma_1)\psi(\Gamma_2), 
\ee
\be
\phi(\Gamma_1\Gamma_2)=\phi(\Gamma_1)\psi(\Gamma_2)+\phi(\Gamma_2)\psi(\Gamma_1), 
\ee
\be
\Phi(\Gamma_1\Gamma_2)=\Phi(\Gamma_1)\psi(\Gamma_2)+\Phi(\Gamma_2)\psi(\Gamma_1).
\ee
Here, the remainders $R_\gamma^\Gamma$, $\tilde{R}_\gamma^\Gamma$, $\bar{R}_\gamma^\Gamma$
are all of higher degrees in the subgraph variables than $\psi(\gamma)$. This is crucial to achieve renormalizability \cite{BrownKreimer}.
 
For two adjacent vertices $x,y\in V_\Gamma$ and $\mathbb{R}\ni r>0$ let us also define
\be 
\phi^r_{x,y}(\Gamma)=\sum_{T_1\cup T_2} (Q(T_1)\cdot Q(T_2))^r \prod_{e\not\in T_1\cup T_2}A_e,
\ee
where 
\be 
(Q(T_1)\cdot Q(T_2))^r=(Q(T_1)\cdot Q(T_2))-r,
\ee 
if $T_1\cup T_2$ separates $x,y$, and
\be 
(Q(T_1)\cdot Q(T_2))^r=(Q(T_1)\cdot Q(T_2)),
\ee
 if it does not (a 2-tree separates two vertices $x,y$ if each of the two trees contains one of them).
\be 
\Phi^r_{x,y}(\Gamma)=\phi^r_{x,y}(\Gamma-b_k)-M(\Gamma)\psi(\Gamma).
\ee

One checks $\Phi^r_{x,y}(\Gamma_1\Gamma_2)=\Phi^r_{x,y}(\Gamma_1)\psi(\Gamma_2)+\Phi^r_{x,y}(\Gamma_2)\psi(\Gamma_1)$
when $x,y$ are adjacent in any one of the two graphs. If $x,y$ are clear from the context, they are often omitted.

\subsubsection{Relations}
Let $x,y$ be two adjacent vertices of $\Gamma$ connected by a $k$-edge $\gamma=b_k$ for some $k$.
Let $r_k$ be the rose $r_k$  obtained by identifying the two vertices $x,y$.

We have
\begin{lem}\label{Lemma1}
\beas
\Phi(\Gamma)
& = & \underbrace{\Phi(\Gamma/\gamma)\psi(\gamma)}_{k-1}
+\underbrace{\Phi(\Gamma-\gamma)\psi(r_k)-M(\gamma)
\psi(\Gamma/\gamma)\psi(\gamma)}_{k}\\
& & -\underbrace{M(\gamma)\psi(\Gamma-\gamma)\psi(r_k)}_{k+1},
\eeas
where we indicate the order in the subgraph variables.
\end{lem}
\begin{proof}Elementary from the definition of graph polynomials $\psi,\phi,\Phi$. Indeed, for $k$-edges as subgraphs one confirms immediately
$R^\Gamma_\gamma=\psi(r_k)$ $\psi(\Gamma-\gamma)$, $\tilde{R}^\Gamma_\gamma=\psi(\gamma)\phi(\Gamma-\gamma)$
and $M(\Gamma)=M(\Gamma-\gamma)+M(\gamma)$. 
\end{proof}

\subsection{Feynman Rules}

\subsubsection{renormalized Feynman rules}

For graphs of a renormalizable field theory, we get renormalized Feynman rules for an overall logarithmically divergent graph $\Gamma$ ($w(\Gamma)=0$) 
with logarithmically divergent  subgraphs as
\be 
\Phi_R=\int_{\mathbb{P}_\Gamma}\sum_{F\in\mathcal{F}_\Gamma}(-1)^{|F|}\frac{\ln\frac{\Phi_{\Gamma/F}\psi_F+\Phi^0_F\psi_{\Gamma/F}}{\Phi^0_{\Gamma/F}\psi_F+\Phi^0_F\psi_{\Gamma/F}}}{\psi^2_{\Gamma/F}\psi^2_F} \Omega_\Gamma.
\ee
Formula for other degrees of divergence for sub- and cographs can be found in \cite{BrownKreimer}. In particular, also overall convergent graphs are covered. 

The Hopf algebra in use in the above is based on the renormalization coproduct $\Delta$.

The antipode $S(\Gamma)$ in this Hopf algebra can be written as a forest sum:
\be
S(\Gamma)=-\Gamma-\sum_{F\in\mathcal{F}_\Gamma}(-1)^{|F|} F\times(\Gamma/F).
\ee
This covers the generic graphs we treat here (of which the graphs apparent in a renormalizable field theory are just a subset), where $\Phi^0$ evaluates using 
renormalization conditions which agree with Lagrangean renormalization conditions whenever $\Phi^0$ acts on a graph $\Gamma\subset H_R\subset H$
so that such conditions are defined.

To renormalize generic graphs we need more general kinematic renormalization conditions. There is freedom in the corresponding choice which shall not concern us here: the analytic structure of graphs is independent of the chosen renormalization point, as long as we stay in the realm of kinematic renormalization schemes.
\subsubsection{Renormalized Feynman rules for cut graphs}
We now give the Feynman rules for a graph with some of its internal edges cut. This can be regarded as giving Feynman rules for a pair $(\Gamma,F)$.
Set $G:=\Gamma/f_{v_\Gamma-r}$, $F=T/f_{v_\Gamma-r}-g_{r-1}$.
\be
\Upsilon_G^F:=\int\left(\Phi_R(G^\prime)\prod_{e\in (G^{\prime\prime}-E_{\Gamma_F})}\frac{1}{P(e)}\prod_{e\in E_{\Gamma_F}}\delta^+(P(e))\right)d^{4|G/G^\prime|}k.\label{Upsilon}
\ee
We use Sweedler's notation for the copoduct provided by $\Delta_c$. 

Note that in this formula $\Phi_R(G^\prime)$ has to stay in the integrand. The internal loops of $G^\prime$ have been integrated out by $\Phi_R$,
but $\Phi_R(G^\prime)$ is still an obvious function of loop momenta apparent in $G/G^\prime$.
The existence of this factorization into integrated subgraphs times cut cographs is a consequence of Lem.(\ref{cccuts}).

\subsubsection{Anomalous thresholds}
Let us come back to a generic graph $\Gamma$. We want to determine anomalous thresholds.

For a multiple-edge $\gamma=b_k(x,y),k\geq 1$ between vertices $x,y\in V_\Gamma$, there is an obvious  bijection  between
spanning trees of $\Gamma/\gamma$ and spanning 
2-trees of $\Gamma-\gamma$ which separate $x,y$.
We set $u=(\sum_ {e\in E_\gamma} m_e)^2$.

We analyse the Landau singularities of $\Gamma$ in terms of $\Gamma/\gamma$, where
$\gamma$ is such a $b_k$ multiple-edge. To completely analyse the graph, we have to consider all possibilities to shrink it multiple-edge after multiple-edge. 

\begin{cor}
\[ 
\Phi(\Gamma-\gamma)E^\gamma_k-M(\gamma)
\psi(\Gamma/\gamma)E^\gamma_{k-1}=\Phi^u(\Gamma-\gamma).
\]
\end{cor}
\begin{remark}
This setting is correct for physical Landau singularities of interest at $A_em_e=A_fm_f$, $e,f\in E_\gamma$.
Else, $u$ has to be modified according to Thm.(\ref{anomalous}).
\end{remark}

For the multiple edge subgraph $\gamma$ we switch to variables $A_i=t_\gamma b_i$, and for the edges in $E_{\Gamma/\gamma}=E_{\Gamma-\gamma}$, we switch to
$A_i=t_{\Gamma/\gamma}a_i$.

The equation for a zero of the second Symanzik polynomial becomes
\beas
\Phi(\Gamma) & = & t_{\Gamma/\gamma}^{e_{\Gamma/\gamma}}\Phi(\Gamma/\gamma)(\{a\})t_\gamma^{k-1}\psi(\gamma)(\{b\})\\
 & & +
t_{\Gamma/\gamma}^{e_{\Gamma/\gamma}-1}t_\gamma^k\Phi^u(\Gamma-\gamma)(\{a\})\\
 & & -t_{\Gamma/\gamma}^{e_{\Gamma/\gamma}-2}t_\gamma^{k+1}M_\gamma(\{b\})
\psi(\gamma)(\{b\})\psi(\Gamma-\gamma)(\{a\})\\
 & = & 0.
\eeas
Divide by $t_{\Gamma/\gamma}^{e_{\Gamma/\gamma-2}}t_\gamma^{k-1}$ to get an
quadratic  equation of the form
\be
-t_{\Gamma/\gamma}^2X(\{a,b\})+t_{\Gamma/\gamma}t_\gamma Y(\{a\}))-t_\gamma^2Z(\{a,b\})) =0.
\ee

We set
\be
X(\{a,b\})=-\Phi(\Gamma/\gamma)(\{a\})\psi(\gamma)(\{b\}), 
\ee
\be
Y(\{a\})=\Phi^u(\Gamma-\gamma)(\{a\}),
\ee
\be
Z(\{a,b\})=M_\gamma(\{b\})
\psi(\gamma)(\{b\})\psi(\Gamma-\gamma)(\{a\}). 
\ee
we note that $X(\{a,b\})$ does not only depend on edge variables but also on masses and momenta
$X(\{a,b\})=X(\{a,b\},s,\{Q,M\})$. It is linear in the channel variable $s$ and we can write
\be 
X(\{a,b\})=X_s(\{a,b\}) s+N(\{a,b\},\{Q,M\}).
\ee
We also have $Y(\{a\})=Y(\{a\},\{Q,M\})$ and $Z(\{a,b\})=Z(\{a,b\},\{M\})$.
The graph $\Gamma/\gamma$ has a Landau singularity $(p_A^{\Gamma/\gamma},p_Q^{\Gamma/\gamma})$. Define $Y_0:=Y(p_A^{\Gamma/\gamma},\{Q,M\})$.

Define also the discriminant
\be 
D:=Y^2+4ZX
\ee
and the function
\be 
s(\{a,b\},\{Q,M\})=\frac{Y^2-4ZN}{4ZX_s}.
\ee

Let $\mathit{T}_s^\Gamma$ be the set of all ordered spanning trees $T$ of a fixed graph $\Gamma$ which allow for the same associated channel variable $s$. 

We have the following result.
\begin{thm}\label{anomalous}
i) A necessary and sufficient condition for a physical Landau  singularity is $Y_0>0$ with $D=0$.\\
ii)  The corresponding anomalous threshold $s_F$ for fixed masses and momenta $\{M,Q\}$ is given as the minimum of $s(\{a,b\},\{Q,M\})$ varied over edge variables $\{a,b\}$. It is finite ($s_F>-\infty$) if the minimum is a point inside $p\in\mathbb{P}^{e_\Gamma-1}$ in the interior of the simplex $\sigma_\Gamma$ (see (\ref{simplex})). If it is on the boundary
of that simplex, $s_F=-\infty$.\\
iii) If for all $T\in \mathit{T}_s^\Gamma$ and for all their forests $(\Gamma,F)$ we have $s_F>-\infty$, the Feynman integral $\Phi_R(\Gamma)(s)$ is real analytic as a function of $s$
for $s<\min_F\{s_F\} $.
\end{thm}  
\begin{proof}
i) follows from the definition of a physical Landau singularity.\\
ii) follows from analysing $\partial_e \frac{Y^2-4ZN}{4ZX_s}=0$, using that the denominator vanishes along the boundaries $A_e=0$.\\
iii) follows as for sufficiently low $s$ as indicated $\Phi_R(\Gamma)(s)$ is real.
\end{proof}
\begin{remark}
We treat massive quantum fields here, $m_e>0,\forall e$. Quadratic equations turn linear, discriminants do not exist and the analysis proceeds differently
in the massless case.
\end{remark}
\section{Cutkosky's theorem}\label{CutkThm}

Our objective in this section is a proof of Cutkosky's theorem (theorem \ref{ct}). We give the proof in the general case subject to the existence of a convenient renormalization scheme. In the case where the cuts leave no remaining loops, the argument does not require renormalization and is much simpler.

\subsection{Thom's First Isotopy Lemma}\label{isotopy}

We want to study the monodromy of a family of unions of propagator quadrics which are affine varieties of the form $(x-p)^2-m^2$ where $x,p\in \A^D$ and the square is a Minkowski or Euclidean quadratic form. This is a ``dirty'' picture in the sense that there are many complicating features which do not play a role and which we will need to filter out. For example, the quadrics have ``trivial'' singularities (e.g. if $m=0$ and $p=0$, the locus $x=0$ is singular.) Further, we will need to compactify, and compactification will introduce new singularities at infinity. The first isotopy lemma of Thom enables us to disguard this ``noise'' and focus on the vanishing cycle which is the phenomenon of physical interest. For more details on this material, the reader is referred to the wonderful notes of Mather \cite{Mather}. 

Let $M$ be a smooth manifold, and let $S\subset M$ be a closed subset which admits a {\it Whitney pre-stratification}. We do not go into detail about Whitney pre-stratifications. The condition that $S$ admit one is very weak. It is always the case, for example, when $M$ is a smooth algebraic variety over $\C$ and $S$ is a (singular) closed subvariety. For complete details, see (op. cit.). Let $P$ be another smooth manifold, and let $f:M \to P$ be a smooth map. 
\begin{thm}[Thom's first isotopy lemma] Suppose $f|S: S \to P$ is proper, and $f|X: X \to P$ is a submersion for every stratum $X \subset S$. Then $(S,f|S,P)$ is a locally trivial bundle. In particular, if $P$ is contractible, then for $p\in P$ we have $S \cong f^{-1}(p)\times P$. 
\end{thm}

We will apply the isotopy lemma as follows. We are given $T \subset M$ closed, admitting a pre-stratification, and a point $m\in T$. Let $B_\ve$ be a small, open ball about $m$ in $M$, and let $S=T-T\cap B_\ve$. Let $p=f(m)$. We replace $P$ by a small contractible neighborhood $U$ of $p$ in $P$ and then we replace $M, T, S$ by their intersections with $f^{-1}(U)$. We can arrange that $S,T$ have Whitney pre-stratifications. Suppose it is the case that $f$ is a submersion on the strata of $S$.  We are interested in the monodromy on the homology of the fibres of $f_T: T \to U$ along a closed path $\phi:[0,1]\to U-\{p\}$. Let $\phi(0)=\phi(1)=q\neq p$. We reduce by pullback to the case where $U\subset \C$ is an open disk and $f: T \to U$ is submersive on strata when restricted over $U-\{p\}$. Applying the isotopy lemma to the pullback of $S\subset T$ over $[0,1]$, we obtain isomorphisms
\eq{}{\begin{CD} f_S^{-1}(\phi(0)) @> \subset >> f_T^{-1}(\phi(0)) \\
@VV \cong V @VV \cong V \\
f_S^{-1}(\phi(1)) @> \subset >> f_T^{-1}(\phi(1)).
\end{CD}
}
For $f_S$, the isotopy lemma applies over the full contractible $U$ which implies that the isomorphism $ f_S^{-1}(\phi(0))\cong f_S^{-1}(\phi(1))$ can be taken to be the identity. (See Proposition (2.4) and Corollaries (2.1) and (2.11) in \cite{Pham3}.) Using this, the variation $var:= \phi -id: H_*(f_T^{-1}(q)) \to H_*(f_T^{-1}(q))$ can be studied locally around $m$. Namely, we have a commutative diagram 
\eq{}{\begin{CD} H_*(f_T^{-1}(q)) @>var >> H_*(f_T^{-1}(q)) \\
@VV \text{excision} V @AA i_* A \\
H_*(f_T^{-1}(q)\cap B_\ve,f_T^{-1}(q)-f_T^{-1}(q)\cap B_\ve) @>var >> H_*(f_T^{-1}(q)\cap B_\ve).
\end{CD}
}

\subsection{Whitney pre-stratifications for graphs}
In this section we give a brief discussion of the Whitney pre-stratifications which underlie the monodromy related to dispersion relations in physics. To a certain extent, the construction is straightforward, only depending on the dimension of space-time and the number of edges of the graph. However, an issue arises at infinity that we do not resolve. The isotopy lemma requires that the map $f$ be proper. In compactifying the Feynman picture, one encounters the blowup of the projectivized first homology of the graph. Strata lying on this blowup will depend in a rather subtle and interesting way on the structure of the graph. (They do not, however, depend on masses or external momenta.) Our study focuses on vanishing cycles which appear at finite distance, and we have not analysed possible exotic classes associated to this structure at infinity. 

Let $G$ be a connected graph with edge set $E$ and vertices $V$. Let $D$ be the dimension of space-time. Write $H:= H_1(G, \C^D)$ for the homology of the graph tensored with $\C^D$. We have the usual exact sequence from topology to which we have added an extra coordinate since we will want to projectivize:
\eq{}{0 \to H \to \C^{DE}\oplus \C\mu \xrightarrow{\partial\oplus 1} (\C^D)^{V,0} \oplus \C\mu \to 0.
}
This gives a rational map on the projective spaces of lines through the origin
\eq{}{f: \P( \C^{DE}\oplus \C\mu) \dasharrow \P((\C^D)^{V,0} \oplus \C\mu).
}
This map is defined off $\P(H) \subset \P( \C^{DE}\oplus \C\mu)$. Let 
$$M=\text{Blow}\Big(\P(H)\subset \P( \C^{DE}\oplus \C\mu)\Big)$$
be the blowup of $\P(H)$, so $f$ extends to a map $g: M \to P:= \P((\C^D)^{V,0} \oplus \C\mu)$. The exceptional divisor $\sE \cong \P(H)\times P \subset M$, and the composition
\eq{}{\P(H)\times P = \sE \subset M \to P
}
is just projection. 

For $e\in E$ we write $e^\vee: \C^{DE}\oplus \C\mu \to \C^D$. Similarly, we write $\mu^\vee: \C^{DE}\oplus \C\mu \to \C$. To each edge $e$ we associate a mass $m_e$. For clarity of exposition, we will assume none of the $m_e$ vanish. We assume $\C^D$ endowed with a non-degenerate quadratic form and we write (abusively) $e^{\vee,2}$ for the evident quadratic form on $ \C^{DE}\oplus \C\mu$. Let $\widetilde Q_e: F_e:= e^{\vee,2}-m_e^2\mu^{\vee,2}$ as a homogeneous quadric on $\P( \C^{DE}\oplus \C\mu)$. We pull $F_e$ back to the blowup $M$ and write $Q_e$ for the corresponding divisor on $M$. Note that $Q_e\cap \sE = (\widetilde Q_e\cap \P(H))\times P$. We write 
\eq{}{Q:= \bigcup_{e\in E} Q_e \subset M.
}

We first consider the Whitney pre-stratification on 
$$Q^0:= Q\cap M^0:=Q\cap (M-\{\mu^\vee=0\})= Q\cap \C^{DE}.$$ 
Let $Q_e^0:e^{\vee,2}-m_e^2=0$ in $\C^D$. Note the $Q_e^0$ are smooth. We stratify $\C^D$ with open stratum $\C^D-Q_e^0$ and closed stratum $Q_e^0$. We view $M^0 = \prod_E\C^D$ as a product of stratified spaces, and give it the product stratification. Then $Q^0 \subset M^0$ is a union of strata, so it inherits a stratification. 

The divisor $\mu=0$ in $M$ has two irreducible components, $\sE$ and the blowup $L$ of $\P(\C^{DE})$ along $\P(H)$. Let $L^0:=L-L\cap \sE = \P(\C^{DE})-\P(H)$. Note that restricting to $L^0$ kills the mass terms. To stratify $Q\cap L^0$ we first stratify $\C^D$ with strata $\{0\}\subset \{e^{\vee}=0\} \subset \{e^{\vee,2}=0\} \subset \C^D$. We then stratify $\prod_E \C^D$
with the product of these stratifications. We remove the smallest stratum $\{0\}\subset\C^{DE}$. The resulting stratification is equivariant for the $\C^\times$-action, so we get a stratification on $\P(\C^{DE})$. Restricting to the open set $L^0 \subset \P(\C^{DE})$ yields the desired stratification.  Again, the structure at $\infty$ is not our primary focus. We simply remark that strata involving the closed stratum $\{e^\vee=0\}$ can be understood from the Whitney pre-stratification of the graph obtained from $G$ by cutting the edge $e$. 

Finally, the stratification of $\sE = \P(H)\times P$ is more subtle. We won't have much to say about monodromy associated to $\sE$ so we merely give an example. Suppose $G=e_1\cup e_2\cup e_3$ is the triangle graph. In this case, $H=\C^D(e_1+e_2+e_3)$. We consider $Q_1\cap Q_2$ on $M-\sE$ and how its closure $Q_{12}$ meets $\sE$. The two equations $e_i^{\vee,2}-m_i^2\mu^{\vee,2},\ i=1,2$ meet $\P(H)=\P(\C^D)$ in the same divisor. The difference $e_1^{\vee,2}-e_2^{\vee,2}-(m_1^2-m_2^2)\mu^{\vee,2}$ lies in the ideal defining $\P(H)$. If we write $e_i^\vee = (x_i^1,\dotsc,x_i^D)$ then this difference lies in the ideal $(x_1^1-x_2^1,\dotsc,x_1^D-x_2^D,\mu^\vee)$. In particular, $Q_{12}$ meets $\sE$ in a family of hyperplanes in $P$ parametrized by the divisor $Q_1\cap \P(H)$. We do not know if this geometry at infinity has any physical consequences.

\subsection{Pham's Vanishing Cycles}

In this section we recall briefly an extension due to F. Pham of the theory of vanishing cycles and the Picard-Lefschetz transformation. More details can be found in \cite{Pham1}, \cite{Pham2}, and \cite{Pham3}.

Let $f: \sX \to \Delta$ be a proper family of varieties over a disk $\Delta \subset \C$. We assume $f$ is smooth except over $0\in \Delta$, and that $f^{-1}(0)$ has a single ordinary double point at $x_0$ and is smooth away from $x_0$. With these hypotheses, we can find coordinates $x_1,\dotsc,x_n$ on $\sX$ near $x_0$ and $t$ on $\Delta$ such that locally the family is given by 
\eq{}{x_1^2+\ldots +x_n^2=t.
}
For $0<\ve<<1$, the homology of a ball $B$ around $x_0$ intersected with $\sX_\ve$ is generated by the class of the real sphere $S_\ve: \sum x_i^2=\ve,\ x_k \in \R$. The monodromy on the smooth fibre $H_*(f^{-1}(\ve),\Q)$ is given by the Picard-Lefschetz formula
\eq{}{c \mapsto c+(-1)^{n(n+1)/2}\langle \text{ex}(c),S_\ve\rangle i_*S_\ve.
}
Here 
$$\text{ex}: H_*(f^{-1}(\ve)) \to H_*\Big(B\cap f^{-1}(\ve),\partial( B\cap f^{-1}(\ve))\Big)
$$
denotes excision and $i_*: H_*(B\cap f^{-1}(\ve)) \to H_*(f^{-1}(\ve))$ is the push-forward. The pairing $\langle \text{ex}(c),S_\ve\rangle$ refers to the natural pairing  
$$ H_{n-1}\Big(B\cap f^{-1}(\ve),\partial( B\cap f^{-1}(\ve))\Big)\otimes H_{n-1}(B\cap f^{-1}(\ve)) \to \Q.
$$

Recall the Feynman amplitude of a graph is an integral of the form $\int_{\R^{Dg}}\frac{d^{Dg}x}{f_1f_2\cdot\cdots \cdot f_n}$. Here the $f_i$ are propagator quadrics which depend on masses and external momenta. The vanishing cycles addressed by Pham's theory occur when some subset of the quadrics $Q_i: f_i=0$ do not meet transversally at a point. In the remainder of this section we develop the theory abstractly, following Pham \cite{Pham1}. 

Let $f:M \to P$ be a smooth, proper map of algebraic varieties. Let $S = \bigcup_{i=1}^n S_i\subset M$ will be a normal crossings divisor. Let $s_0\in \bigcap_{i=1}^nS_i \subset S\subset M$, and let $p_0=f(s_0)$. We will work locally near $s_0$ on $M$, so we may fix analytic coordinates $t_1,\dotsc,t_k$ around $p_0$ on $P$ and $x_1,\dotsc,x_\ell$ around $s_0$ on $M$ in such a way that $t_1,\dotsc,t_k,x_1,\dotsc,x_\ell$ form a full set of coordinates on $M$ near $s_0$. We assume that locally 
\ga{}{S_i: x_i=0;\quad i=1,\dotsc,n-1, \\
S_n: t_1-(x_1+\cdots+x_{n-1}+x_n^2+\cdots+x_\ell^2)=0. 
}
For such a configuration of divisors, the point $s_0$ (origin) is called a {\it pinch point}. Notice that on the smallest stratum $\bigcap_1^n S_i$ of $S$, viewed as a variety fibred over $P$, we have a classical Picard-Lefschetz vanishing cycle local equation, which is a family of spheres $x_n^2+\cdots+x_\ell^2 = t_1$ degenerating as $t\to 0$. For $t_1=\ve>0$, the {\it vanishing sphere} $\rm{vsphere}_\ve$ is the real sphere 
\eq{vs}{\rm{vsphere}_\ve: \ve=x_n^2+\cdots+x_\ell^2;\quad x_i\in \R,\ n\le i\le \ell. 
}
Classically, $\rm{vsphere}_\ve$ is referred to as the vanishing cycle, but we follow Pham here and distinguish $3$ different topological chains, the vanishing sphere, the vanishing cell, and the vanishing cycle. 
\begin{defn}With notation as above, the vanishing cell, 
\eq{}{\rm{vcell}_\ve:x_i \ge 0,\ 1\le i\le n-1;\quad \ve-( x_n^2+\cdots+x_\ell^2) \ge 0.
}
The vanishing cycle $\rm{vcycle}_\ve$ is the iterated tube 
$$\tau_{1,\ve_1}^*\tau_{2,\ve_2}^*\cdots\tau_{n-1,\ve_{n-1}}^*(\rm{vsphere}_\ve).
$$
Here $\ve>>\ve_{n-1}>>\cdots >>\ve_1>0$. The notation $\tau_{i,\ve_i}^*$ refers to pulling back to the circle bundle of radius $\ve_i$ inside the (metrized) normal bundle for $S_1\cap\cdots\cap S_i \subset S_1\cap \cdots \cap S_{i-1}$. This circle bundle is viewed as embedded in $S_1\cap\cdots\cap S_{i-1}$ and not meeting $S_1\cap\cdots\cap S_{i}$. The inequalities $\ve_i<<\ve_{i+1}$ insure that $\rm{vcycle}_\ve$ is a closed chain on $M-\bigcup_{i=1}^n S_i$. 
\end{defn}

The local structure of homology around $s_0$ are computed by Pham to be 
\begin{thm}\label{homology}Let $B$ be a small ball around $s_0$ in $M$. Let $p\in P$ be near $p_0$ and assume $t_1(p)=\ve>0$. Write $d=\dim_\C M_p$.
\nn
(i) We have for reduced homology 
$$\widetilde H_j\Big(\bigcap_1^nS_{i,p}\cap B,\Z\Big)=(0),\ j\neq d-n:= \dim_\C \Big(\bigcap_1^nS_{i,p}\Big).
$$
$$\widetilde{H}_{d-n}\Big(\bigcap_1^nS_{i,p}\cap B,\Z\Big)=\Z\cdot \rm{vsphere}_\ve.
$$
\noindent (ii) In relative homology
$$H_j(B_p,S_p\cap B) = (0), \ j\neq d;\quad H_{d}(B_p,S_p\cap B) =\Z\cdot \rm{vcell}_\ve.
$$
\noindent (iii) For the homology of the complement
$$H_j(B_p-S_p\cap B) = (0),\ j\neq d;\quad H_{d}(B_p-S_p\cap B) =\Z\cdot \rm{vcycle}_\ve.
$$
\end{thm}

Local Poincar\'e duality yields a pairing
\eq{}{\langle\cdot,\cdot\rangle: H_d(B_p-S_p\cap B)\otimes H_{d}(B_p,(S_p\cap B)\cup \partial B_p) \to H_{2d}(B_p,\partial B_p) \cong \Z.
}
Let $\gamma$ be a simple closed path on $P$ based at $p$, supported in a neighborhood of $p_0$, and looping once around the divisor $t_1=0$. Using the isotopy lemma as discussed in section \ref{isotopy}, and assuming that $f$ is a submersion on strata except at the point $s_0$, Pham shows that the variation of monodromy
\eq{}{var:= \gamma_*-Id:H_*(M_p-S_p)\to H_*(M_p-S_p)
}
factors through excision and a local variation map $var_{loc}$
\ml{}{ H_d(M_p-S_p) \xrightarrow{excision} H_d(B_p-S_p\cap B_p,\partial B_p) \xrightarrow{var_{loc}}\\
 H_d(B_p-S_p) \xrightarrow{i_*} H_d(M_p-S_p).
}
(The variation map is zero in homological degrees $\neq d$.) 

The Picard-Lefschetz theorem in this setup is
\begin{thm}\label{pl_thm} We have
\ml{}{var_{loc}(\rm{excision}(c)) = \rm{excision}(c) \\
+(-1)^{(n+1)(n+2)/2}\langle \rm{excision}(c),\rm{vcell}_\ve\rangle \rm{vcycle}_\ve
}
\end{thm}
\begin{proof} See \cite{Pham1}, \cite{Pham2}. 
\end{proof}

\begin{ex}In this example we compute the Picard-Lefschetz transformation arising in physics. One has a set of quadrics $Q_i:f_i=0,\ 1\le i\le n$ indexed by the edges of a graph. Our normal crossings divisor in this case is $\bigcup Q_i$. We consider the situation locally around a pinch point $s_0$, and we write in Pham's coordinates, $Q_i: x_i=0,\ 1\le i\le n-1$ and $Q_n: t_1-(x_1+\cdots +x_{n-1} + \sum_n^\ell x_i^2)=0$. Notice that in these coordinates, the bad fibre over $p_0$ (where $t_1=0$) of $\bigcap_1^n Q_{i,p_0}$ has an isolated $\R$-point at the origin. We will see (proposition \ref{real}) in the case of {\it physical singularities} that this is also the case for the space-time coordinates. Hence it is plausible to assume that at least for physical singularities Pham's local coordinates are defined over $\R$. We are interested in computing the monodromy on $M-\bigcup Q_i$ of the chain given by taking space-time coordinates in $\R$. As it stands this makes no sense, because our quadrics are defined using the Minkowski metric so the $Q_i$ meet the real locus. The standard ploy to avoid this problem is to replace the defining equations $f_j=0$ by $f_j+ia_j=0$ with $1>>a_j>0$.  Using theorem \ref{pl_thm}, we need to compute $\langle \rm{excision}(\R^{Dg}),\rm{vcell_\ve}\rangle$, where $D$ is the dimension of space-time and $g$ is the number of loops of the graph. 

In Pham's local coordinates, the deformed quadrics are defined by $x_j=-ia_j,\ 1\le j\le n-1$ and $\ve + ia_n - (\sum_1^{n-1}x_j+\sum_n^{Dg} x_j^2)=0$. We let $x_j$ run over the path $x_j=(1-\rho_j)ia_j + \rho_j(\ve+ia_n)/(n-1),\ 1\le j\le n-1$. Here $0\le \rho_j\le 1$. For points $x^0=(x_1^0,\dotsc,x^0_{n-1})$ on those paths, the quadratic equation becomes $\sum_n^{Dg} x_j^2 = \ve + ia_n - \sum_1^{n-1}x_j^0=:r(x^0)e^{i\theta(x^0)}$. The locus
\eq{17}{\{(e^{i\theta(x^0)/2}u_n,\dotsc e^{i\theta(x^0)/2}u_{Dg})\ |\ \sum_{k=n}^{Dg}u^2_k \le r(x^0)\}
} 
is a solid sphere, and the union of those solid spheres as the $x_j$ run over those paths is the chain $\text{vcell}_\ve$ for the deformed quadrics. 

According to theorem \ref{pl_thm}, the multiplicity of the vanishing cycle in the variation of the cycle $\R^{Dg}$ (or more precisely of $\P^{Dg}(\R)$) around the given pinch point is the intersection of this real locus with $\text{vcell}_\ve$.

Note that $x_j$ is real only at the point $\rho_j=a_j/(a_j+a_n/(n-1))\in (0,1)$. At the point with these coordinates, the quadratic equation reads $\sum_n^{Dg} x_j^2 = R+ia_n$ for some $R\in \R$. since $a_n\neq 0$, the only real point on the corresponding solid sphere \eqref{17} lies at the origin. Thus, $\text{vcell}_\ve\cap \R^{Dg}$ is precisely one point. The intersection looks locally like the intersection $\R^N\cap e^{i\theta}\R^N \subset \C^N$, so it is transverse, and $$\langle\rm{excision}(\P^{Dg}(\R)),\rm{vcycle}\rangle = \pm 1.
$$ 
We do not try to compute the sign. 
\end{ex}

\subsection{Cutkosky Rules}

Cutkosky rules \cite{Todorov,IZ}, \cite{SMatrix}, give a formula for the variation of an amplitude integral around a threshold point in the space of external momenta. The formula as classically expressed (\cite{SMatrix}, formula (2.9.10)) is
\eq{18}{\text{var}(I) = (-2\pi i)^r\int\frac{\delta^+(q^2_1-m_1^2)\cdots \delta^+(q^2_r-m_r^2)d^{Dg}x}{(q^2_{r+1}-m_{r+1}^2)\cdots (q^2_N-m_N^2)}.
}
Here $I$ is the amplitude associated to a graph $G$ with $N$ edges and $g$ loops. The $q^2_i-m_i^2$ are propagators, and $\delta^+(q^2_i-m_i^2)$ means to take the residue of the original form $\omega := \frac{d^{Dg}x}{\prod_1^N(q^2_i-m_i^2)}$ along the divisor $q^2_i-m_i^2=0$ and then to integrate over the part of the real locus of that divisor where the energy is positive. External momenta are placed at a point near the threshold divisor associated to a pinch point of the intersection $\bigcap_{i=1}^r\{q^2_i-m_i^2=0\}$, and ``var'' (also called ``discontinuity'') refers to the variation of $I$ as external momenta wind around the threshold divisor. 

The calculus of Pham's vanishing cycles leads to a rigorous proof of the Cutkosky formula in cases when it holds. The crucial point to be understood is the role of the real structure. Note the real structure plays a central role in the physical narrative; $\delta^+(q^2_i-m_i^2)$ puts the edge ``on shell''. The issue is under what conditions \eqref{18} is the integral of the residue of $\omega$ over Pham's vanishing sphere on $\bigcap_{i=1}^r\{q^2_i-m_i^2=0\}$. In this section, we prove this in the basic case $r=N$ at a {\it physical pinch point}. 

For convenience we write $\ell_i = q^2_i-m_i^2$. We consider the universal quadric $\ell:= a_1\ell_1+\cdots +a_N\ell_N$ where the $a_i$ are variables. Let $Q\subset \C^N\times \C^{Dg}$ be the zeroes of the universal quadric. 
\begin{defn}\label{physical} A point $p=(c_1,\dotsc,c_N,p') \in Q(\R)$ is said to be a physical singularity if $p$ is a singular point on $Q$ and all the $c_i>0$. By extension, a point $p'\in \R^{Dg}$ is called a physical singularity if there exists a physical singularity of the form $p=(c_1,\dotsc,c_N,p')$. 
\end{defn}
\begin{lem}Let $p=(c_1,\dotsc,c_N,p')$ be a physical singularity of the universal quadric $Q$. Then $p'\in \bigcap_{i=1}^N \{\ell_i=0\}$ and $p'$ is singular on the intersection. 
\end{lem}
\begin{proof}By assumption, $p$ is singular on $Q$, so $\ell_i(p') = \frac{\partial}{\partial a_i}\ell(p) = 0$ and $p'$ lies in the intersection of the quadrics. Also, vanishing of the partials of $\ell$ with respect to coordinates on $\R^{Dg}$ yield the relation
\eq{19}{\sum_{i=1}^N c_i\text{grad}_{p'}(\ell_i) = 0
}
which implies that $p'$ is singular on the intersection of the quadrics. 
\end{proof}

The following is the analogue in this setup of Pham's pinchpoint condition \cite{Pham1}. 
\begin{defn} A physical singularity $p$ will be called non-degenerate if two conditions hold\nn
(i) The relation \eqref{19} among the gradients of the $\ell_i$ at $p'$ is the only non-trivial linear relation. \nn
(ii) The Hessian matrix at the singular point for the intersection of the quadrics is non-degenerate. 
\end{defn}

Condition (i) in the above definition means we can choose local coordinates $x_1,\dotsc,x_{N-1}$ at the singular point in such a way that the intersection of the quadrics is cut out locally near the singular point by 
\eq{20}{x_1=\cdots =x_{N-1}=g=0
}
for some function $g$. Condition (ii) says that the matrix of second order partials of $g|\{x_1=\cdots =x_{N-1}=0\}$ at the point $p'$ is non-degenerate. 

In order to understand the $\R$-structure we borrow a standard picture of a Morse function $f$, \cite{Morse}. 
$$\morse$$
The hessian of $f$ at the points $a, d$ is definite, while the hessian at $b,c$ has one positive and one negative eigenvalue. Note $f^{-1}f(a)=a$ and the fibre $f^{-1}(x)$ for $x$ slightly above $f(a)$ is a sphere (circle in this case) which contracts to a point as $x\to f(a)$. The picture is similar at $d$.  On the other hand $f^{-1}f(b)$ and $f^{-1}f(c)$ are  figure-eights and there are no vanishing cycles. In our case, if we know that the hessian of the function $g$ in \eqref{20} at the singular point is (either positive or negative) definite, then Pham's vanishing cycle will be the real sphere given to us by Morse theory.

\begin{prop}\label{real} Let $(c_1,\dotsc,c_N,p') \in Q$ be a non-degenerate physical singularity (definition \ref{physical}). Assume the squares of the masses $m_i^2>0$ for $1\le i\le N$. Then the hessian associated to $p' \in \bigcap_1^N \{\ell_i=0\}$ is negative definite. In particular, the vanishing cycle is real. 
\end{prop}
\begin{proof}Order the edges $e_1,\dotsc,e_N$ such that $q_1,\dotsc,q_g$ give coordinates on the space of loops. We view $q_i$ for $i>g$ as a linear combination of $q_1,\dotsc,q_g$. (Note each $q_i = (q_i^1,\dotsc,q_i^D)$ is a $D$-tuple of coordinate functions.)  It will be convenient to shift to the more familiar physics viewpoint and think of $p=(c_1,\dotsc,c_N,q_1(p'),\dotsc,q_g(p'))\in \R^N\times \R^{Dg}$. We then have $\ell_i = (q_i-s_i)^2-m_i^2$ for some $s_i \in \R^D$. Because $p$ is singular on $Q$, the sum $\sum_1^N c_i\ell_i$ is pure quadratic when expanded about $p'$, with vanishing constant and linear terms. Notice that for all $i\le N$, the quadratic part of $\ell_i$ is $q_i^2$ which is a Minkowski square. The quadratic form associated with the hessian is then simply $\sum_1^N c_iq_i^2$. Of course, these are Minkowski squares so it is not possible to say anything about the sign of the form on $\R^{Dg}$. However, what we want is the sign of the form restricted to the intersection of the tangent spaces at $p'$ of the $\ell_i$. Among these are the tangent spaces at $p'$ to $\ell_i = q_i^2-m_i^2$ for $1\le i\le g$. The tangent space at $p'$ for such an $\ell_i$ is $q_i^{-1}(q_i(p')^\perp) \subset \R^{Dg}$. Here $q_i(p')^\perp \subset \R^D$. Note $q_i(p')^2=m_i^2>0$ for $1\le i\le g$. Since the Minkowski metric has sign $(1,-1,\dotsc,-1)$ on $\R^D$, we conclude that the sign of the metric on the intersection of these spaces in $R^{Dg}$ is negative definite. In particular, $\sum_1^N c_iq_i^2$ is certainly negative semi-definite on the intersection of these tangent spaces. The only way it could fail to be negative definite would be if there was a non-zero tangent vector $t$ such that $q_i(t)=0,\ 1\le i\le N$. But this is not possible since the $q_i,\ 1\le i\le g$ give coordinates. 
\end{proof}

\subsection{Cutkosky cuts}
We collect together in this section some basic remarks about Cutkosky cuts. $G$ will denote a connected graph with edge set $E$ and vertex set $V$. Each vertex $v$ carries an external momentum $p_v \in \C^D$ where $D$ is the dimension of space-time, and each edge $e$ has associated a mass $m_e\in \C$. Physically, of course, one is interested in real masses and real external momenta, but we will need to work with complex values as well. We view $\R^D$ as carrying a Minkowski metric with sign $(1,-1,\dotsc,-1)$. This quadratic form is extended in an obvious way to $\C^D$, and we write $p_v^2$ for the metric square of an external momentum $p_v\in \C^D$. 

The total momentum $p:= \{p_v\ |\ v\in V\}$ is subject to the conservation law $\sum_{v\in V} p_v=0$. Associated to each edge we associate an affine quadric as follows. The homology of the graph with coefficients in $\C^D$ is calculated by the homology sequence
\eq{}{0 \to H_1(G, \C^D) \to (\C^D)^E \xrightarrow{\partial} (\C^D)^{V,0} \to 0
}
where $\partial$ is the usual boundary map on topological chains, and $(\C^D)^{V,0}:=\ker((\C^D)^{V} \xrightarrow{\rm{sum}} \C^D)$. We view the total momentum $p\in (\C^D)^{V,0}$, and we consider the fibre $\partial^{-1}(p) \subset (\C^D)^E$ as an affine torsor under $H_1(G, \C^D)\cong \C^{Dg}$ where $g$ is the loop number of $G$. To each edge $e$ we associate first the evident projection $e^\vee: (\C^D)^E \to \C^D$, and then, by composing with the Minkowski form and restricting, a quadratic function $e^{\vee,2}$ on $\partial^{-1}(p)$. The propagator quadric
\eq{}{\ell_e:= e^{\vee,2}-m_e^2.
}
Let $Q_e:\ell_e=0$. If $m_e\neq 0$, $Q_e$ is a non-singular affine quadric in $\partial^{-1}(p)\cong \A^{Dg}$. We will be interested in subsets $E'\subset E$ of edges such that $\bigcap_{e\in E'} Q_e$ is singular. The following elementary observation is basic:
\begin{prop}Suppose $\bigcap_{e\in E'} Q_e$ is singular and the masses $m_e,\ e\in E'$ are non-zero. Then the graph $G' \subset G$ obtained by cutting the edges $e\in E'$ is not connected. (Note ``cutting" an edge $e$ means removing $e$ but not the vertices of $e$. The cut graph is permitted to have isolated vertices.)
\end{prop}
\begin{proof}Cutting an edge of $G$ removes a $1$-simplex from $G$ so it increases the Euler characteristic $\chi(G) = h^0(G)-h^1(G)$ by $1$. Writing $G-e$ for the cut graph, we see there are two (mutually exclusive) possibilities. either $h_1(G-e)=h_1(G)-1$ or $h_0(G-e)=h_0(G)+1$. It is straightforward to check that $h_1(G-e)=h_1(G)-1$ if and only if the composition
\eq{}{ H_1(G,\C^D) \subset (\C^D)^E \xrightarrow{e^\vee} \C^D
}
is surjective. (Note the kernel of this arrow is $H_1(G-e,\C^D)$.)
When we iterate the argument, we see that $G-\bigcup_{e\in E'} e$ is connected if and only if $\prod_{e\in E'} e^\vee: H_1(G,\C^D) \to (\C^D)^{E'}$ is surjective. But when this happens, the $e^\vee$ give a (partial) set of coordinates on $\partial^{-1}(p)$, and it is clear that the intersection over $e\in E'$ of the quadrics $Q_e$ cannot be singular. 
\end{proof}

Note, however, that the converse is not true. A series of cuts which disconnects a graph does not necessarily correspond to a singular intersection of quadrics. Given $E'\subset E(G)$ such that cutting edges in $E'$ disconnects the graph, one obtains a family of intersections of quadrics parametrized by external momenta. Typically, there will be a discriminant, or threshold, divisor in the space of external momenta such that the intersections become singular for momenta on the divisor. The classical Picard-Lefschetz theory (as reworked by F. Pham, \cite{Pham1}) calculates the monodromy when external momenta wind around that discriminant path in a generic way. In the case of a one parameter family parametrized by a disk with a singular intersection at $t=0$, generic means that the singular intersection of propagator quadrics has an isolated ordinary double point and no other singularity. 

This isolated double point condition has an important consequence for the graph $G$. 
\begin{prop}Let $E'' \subset E(G)$ and assume $G-\bigcup_{e\in E''} e$ is disconnected. If the loop number $h_1(G-\bigcup_{e\in E''} e)>0$, then $\bigcap_{e\in E''}Q_e$ cannot have an isolated singularity.
\end{prop}
\begin{proof}Let $G' = G-\bigcup_{e\in E''} e \subset G$ and define $G'' = \bigcup_{e\in E''}e$. We view $G''$ as a quotient $G \surj G''$ obtained by contracting the edges in $G'$. The propagator quadrics associated to edges in $E''$ are constant on the fibres of $H_1(G,\C^D) \to H_1(G'', \C^D)$. These fibres are identified with $H_1(G', \C^D)$, so if the loop number of $G'$ is greater than zero, the intersection of these quadrics cannot have an isolated singularity. 
\end{proof}

\subsection{Cutkosky's Theorem}
We can now formulate Cutkosky's theorem. We assume given a graph $G$ and a subgraph $G'\subset G$. We write $G'':= G/\!/G'$ for the graph obtained by contracting the connected components of $G'$ to (separate) points. We assume $G''$ has no self-loops.  We identify the edge set of $G''$ as a subset of the edges of $G$, $E'' \subset E$. 

For the proof of Cutkosky's theorem below, we suggest that the reader first consider the case when $H_1(G')=(0)$, so no loops are contracted. In this case, the integration chain $C$, \eqref{26}, is a sphere and the argument is much simpler. We have given the argument in the general case assuming the integral  
\eqref{integrand} is conveniently renormalized.

\begin{thm}[Cutkosky]\label{ct} Assume the quotient graph $G''$ has a non-degenerate physical singularity (definition \ref{physical}) at an external momentum point $p'' \in (\bigoplus_{V''}\R^D)^0$, i.e. the intersection $\bigcap_{e\in E''}Q_e$ of the propagator quadrics associated to edges in $E''$ has such a singularity at a point lying over $p''$. Let $p \in (\bigoplus_{V}\R^D)^0$ be an external momentum point for $G$ lying over $p''$. Then the variation of the amplitude $I(G)$ around $p$ is given by Cutkosky's formula
\eq{}{\text{var}(I(G)) = (-2\pi i)^{\# E''}\int\frac{\prod_{e\in E''} \delta^+(\ell_e)}{\prod_{e\in E'} \ell_e}.
}
\end{thm}
\begin{proof} As we have seen in proposition \ref{real},  near $p''$ we can find a real vanishing sphere $\rm{vsphere_{\ve''}}$ which is the real locus underlying $\prod_{e\in E''} \delta^+(\ell_e)$. Consider the diagram 
\eq{25}{\begin{CD} @. 0 @. 0 @. 0 \\
@. @VVV @VVV @VVV \\
0 @>>> H' @>>> H @>>> H'' @>>> 0 \\
@. @VVV @VVV @VVV \\
0 @>>> (\R^D)^{E'}@>>> (\R^D)^{E}@>\pi_E >> (\R^D)^{E''} @>>> 0 \\
@. @VVV @VV\partial V @VV\partial'' V \\
0 @>>> (\R^D)^{V',0,0} @>>> (\R^D)^{V,0} @>\pi_V >> (\R^D)^{V'',0} @>>> 0 \\
@. @VVV @VVV @VVV \\
@. 0 @. 0 @. 0
\end{CD}
}
Here $H', H, H''$ are the first homology groups of $G', G, G''$ with $\R^D$-coefficients. The superscript $(V',0,0)$ on the lower left means the coefficients of the vertices sum to $0$ over each connected component of $G'$. We are given $\ve \in (\R^D)^{V,0}$, and $\ve'' = \pi_V(\ve)$ is near a threshold point for the quotient graph $G''$ corresponding to a physical singularity (definition \ref{physical}), so the vanishing sphere \eqref{vs} $\rm{vsphere}_{\ve''}\subset \partial''{}^{-1}(\ve'')\cap\bigcap_{e\in E''} Q_e(\R)$ is defined. We take as integration chain (although plausible, some justification is needed here. See below.) 
\eq{26}{C=\pi_E^{-1}(\rm{vsphere}_{\ve''} \cap \partial^{-1}(\ve)).
}
Note $C$ is fibred over $\rm{vsphere}_{\ve''}$ with fibres $H'$. We integrate over $C$ using Fubini. The integrand can be decomposed
\eq{integrand}{\Big(\frac{d^{Dg'}x'}{\prod_{E'}f_e}\Big)\wedge\pi_E^*\Big(\frac{d^{Dg''}x''}{\prod_{E''}f_e}\Big)
}
In this process, the propagators corresponding to edges in $E''$ are left with masses in $\R$ (as they must be in order to contain a real sphere), while propagators from edges in $E'$ have $i\ve$ added to the masses. As a consequence, the poles in the integrand associated to propagators in $E'$ meet the chain of integration only (possibly) at infinity. These singularities may need to be renormalized. 

Finally, to justify our choice of integration chain, we have to show that the Picard-Lefschetz-Pham formula for the variation lifts to the cone with fibres which are $H'$-torsors minus the propagator quadrics associated to edges of $G'$. Let $p'' \in (\R^D)^{V'',0} \subset (\C^D)^{V'',0}$ be a point on a threshold divisor corresponding to a non-degeneralte physical singularity. Let $b'' \in \bigcap_{e\in E''} Q_e(\R)$ be the corresponding singular point, so $\partial''(b'')=p''$. Now fix neighborhoods $b''\in B'' \subset (\C^D)^{E''}$ and $p'' \in P'' \subset (\C^D)^{V'',0}$ as in section \ref{isotopy}. I.e. we assume that the isotopy lemma applies in the complement of $B''$ so monodromy around a small circle around $p''$ in $P''$ can be calculated inside $B''$.  

We fix a linear splitting $\sigma: (\C^D)^{V'',0} \to (\C^D)^{V,0}$ for the projection $\pi_V$ in \eqref{25}. We assume $\sigma$ is defined over $\R$. Let $P=\sigma(P'') \subset (\C^D)^{V,0}$. We want to compute the variation for a loop in $P$ winding around $p:= \sigma(p'')$. Since the local singularity comes from the edges $E''$ which pinch in $U'':= B''\cap \delta''{}^{-1}(P'')$, we will work in
\eq{}{U:= \pi_E^{-1}(U'')\cap \partial^{-1}(P) \subset (\C^D)^{E}. 
}
We have the diagram ($Q_e$ are propagator quadrics)
\eq{}{\begin{CD}U\cap\bigcap_{e\in E'} Q_e  @>>>U@>\pi_{U,E}>> U'' \\
@. @VV\partial V @VV\partial'' V \\
@. P @>\cong >> P'' 
\end{CD}
}

Here again we need to apply our renormalization hypothesis to compactify the fibres of $\pi_{U,E}$ in such a way that the isotopy lemma applies. In particular, we need that the projection map $\pi_{U,E}$ is submersive on the strata of $U\cap\bigcap_{e\in E'} Q_e$. Granting this, we apply the isotopy lemma. Since $U''$ is contractible, we deduce a stratified isomorphism of spaces over $U''$
\eq{30}{U \cong F\times U'';\quad F:= \pi_{U,E}^{-1} (b''). 
}
We will need a certain compatibility with the $\R$-structure. Since the masses of propagators in $E'$ are not real, we see that  for $\beta''\in U''(\R)$, $\pi_{U,E}^{-1} (\beta'')(\R)$ does not meet any of the $Q_e,\ e\in E'$. Said another way, $U(\R) = (U-\bigcap_{e\in E'}Q_e)(\R)$. We want our stratified isomorphism to preserve the real subset. This is possible since we can enlarge the stratification including $U(\R)$. 

 Since $b''$ is an $\R$-point, $F$ is defined over $\R$, and $F(\R)$ is a torsor for $H_1(G', \R^D)$. 

The variation is calculated by replacing $U''$ with $W'':= U''-U''\cap \bigcap_{e\in E''} Q_e$ and deforming $F(\R)\times (W''(\R)\cap\partial''{}^{-1}(\ve''))$ as $e^{i\theta}\ve''$ loops around $p''$. (Here $\ve''\in P''(\R)$ and $e^{i\theta}\ve''$ is intended to suggest a circle around $p''$.) The variation, of course, does not preserve the $\R$-structure. We know by Pham, that winding around returns $W''(\R)\cap\partial''{}^{-1}(\ve'')$ to a chain homologous to  $W''(\R)\cap\partial''{}^{-1}(\ve'') + K\cdot\rm{vcycle}$ where $\rm{vcycle}$ is an iterated tube over $\rm{vsphere}$ and $K$ is constant.  Upstairs in $U$ we identify $U\cong F\times U''$ using \eqref{30} and map chains $c$ on $U''$ to $F(\R)\times c$ on $F\times U'' \cong U$. The problem is that is is hard to say what the corresponding chains on $U$ are. However, compatibility with the $\R$-structure implies that if $c$ is supported on $U''(\R)$ then the chain on $U$ can be taken to be $\pi_{U,E}^{-1}(c)(\R)$. In particular, if $c=\rm{vsphere} \subset \delta''{}^{-1}(\ve'')(\R)$, then the corresponding chain on $U$ is exactly the chain $C$ in \eqref{26}. Since monodromy and taking iterated tubes clearly commute with taking the product with $F(\R)$ (because for this we are working on $F\times U''$; this is the beauty of the isotopy lemma) we conclude that integration over $C$ calculates the variation.  

\end{proof}
\section{Variation and Hodge Weights}

One would like to better understand the distinction between the variation associated to a minimal (or Cutkosky) cut, i.e. a cut which is minimal separating the graph into $2$ pieces, and an anomalous cut which involves cutting more than the minimal set of edges. A simple remark which could be useful for such a study is that the Hodge weight of the vanishing cycles decrease as more edges are cut. 

We consider a collection of quadrics $Q_{e,p}:=\{\ell_{e,p}=0\} \subset \P^{Dg}$ for $e\in E(G)$, the edge set of our graph $G$. Here $p$ denotes external momenta. Write
\eq{}{X_p:=\bigcup_{e\in E} Q_{e,p}.
}
Let $\omega_p = \frac{dx^{Dg}}{\prod_e \ell_{e,p}}$. It may happen that $\omega_p$ has a pole on the hyperplane $H_\infty$ at infinity , in which case we replace $X_p$ by $X_p\cup H_\infty$. In this way, the amplitude becomes a period for the mixed Hodge structure $H^{Dg}(\P^{Dg}-X_p,\Q)$. We write $U_p:=\P^{Dg}-X_p$ and we vary $p$ in the complement of the threshold divisors in the space of external momenta. 

The localization sequence for Betti cohomology yields a long-exact sequence
\eq{}{H^{Dg}(\P^{Dg}) \to H^{Dg}(U_p) \to H_{Dg-1}(X_p)(-Dg) \to H^{Dg+1}(\P^{Dg})
}
We abreviate this 
\eq{}{H^{Dg}(U_p)^0(Dg) \cong H_{Dg-1}(X_p)^0
}
Now we cut edges $e_1,\dotsc,e_r$. We assume the external momentum is near a threshold divisor for this cut, so we may identify a vanishing sphere $\text{vsphere}(p)\in H_{Dg-r}(Q_{e_1,p}\cap \cdots \cap Q_{e_r,p})$. For $p$ sufficiently close to the threshold divisor, the vanishing cycle will be supported in the smooth locus of  $Q_{e_1,p}\cap \cdots \cap Q_{e_r,p}$. Writing $Z\to Q_{e_1,p}\cap \cdots \cap Q_{e_r,p}$ for some resolution of singularities, we may assume the vanishing sphere comes from a class in $H_{Dg-r}(Z)$. It follows that 
\eq{}{\text{vsphere}(p)\in W_{r-Dg}H_{Dg-r}(Q_{e_1,p}\cap \cdots \cap Q_{e_r,p})
}
where $W_{r-Dg}$ denotes the subspace of Hodge weight $r-Dg$. Note the weight filtration is increasing, and $r-Dg$ is the smallest non-trivial weight. 

The iterated tube $\text{vcycle}(p) \in H_{Dg}(U_p)(-r)$. To understand in Hodge-theoretic terms the passage from $\text{vsphere}(p)$ to $\text{vcycle}(p)$ we proceed as follows. We fix $p$. The sphere is contained in the smooth locus of $Q_{e_1,p}\cap \cdots \cap Q_{e_r,p}$. Blowing up on $\P^{Dg}$ away from  $\text{vsphere}(p)$ we can suppose that the $X_i:= Q_{e_i,p}$ are smooth divisors on a smooth variety $P$, and that the $X_i$ meet transversally. Suppose for example $r=2$, and write $Y=X_1\cap X_2$. Since $X_1-Y\inj P-X_2$ is a smooth divisor, we can consider the tube maps dual to residue in cohomology
\ml{}{H_{i-2}(Y) \xrightarrow{\text{tube}} H_{i+1}(X_1-Y)(1) \\\xrightarrow{\text{tube}} H_{i+2}((P-X_2)-X_1)(2) = H_{i+2}(P-(X_2\cup X_1))(2).
}
It follows that
\eq{}{\text{vcycle}(p) \in W_{r-Dg}\Big(H_{Dg}(U_p)(r)\Big) \cong  W_{-r-Dg}H_{Dg}(U_p).
}
We conclude
\begin{prop} The variation associated with cutting $r$ edges maps 
\eq{}{H_{Dg}(U_p) \xrightarrow{var} W_{-r-Dg}H_{Dg}(U_p).
}
Note that as we cut more edges, $r$ increases and $W_{-r-Dg}$ decreases. (The weight filtration is increasing, so $W_{-s} \subset W_{-r}$ for $s>r$.)
\end{prop}

\section{Monodromy vs Dispersion}

Now let $(\Gamma,T)$ be given, and with it a corresponding sequence of forests $F$, cut graphs $\Gamma^F$,
reduced graphs $\Gamma_F$, and sets of graphs $\Gamma_i\equiv \Gamma_i(F)\in G^F$. We let $F_2$ be the 2-forest in this sequence of forests.

The corresponding reduced graph $\Gamma_{F_2}$ is a $b_k$ multiple edge, $k={e_{\Gamma_{F_2}}}$ with $|\Gamma_{F_2}|=k-1$.
$\epsilon_2$ is a cut through $k$ edges. $\Gamma_{F_2}\cap T$ is a single edge, the first edge of $T$ in the ordering of edges of $T$.

If and only if those $\Gamma_i$ contain loops the functions $\Phi_R(\Gamma_i)$ are multivalued and have monodromy in the variation
of their external momenta. 

\begin{remark}
Monodromies from varying only the channel variable $s$ associated to $(\Gamma,T)$, $s\in Q_{\Gamma_{F_2}}\subseteq Q_\Gamma$
(see (\ref{qgamma}))
are determined from threshold divisors $s_{j-2}(\epsilon_j)$.
\end{remark}

\begin{remark}
All edges in $e$ of the reduced graph $\Gamma_{F_2}$ such that $e\cap T=\emptyset$ can be assigned a marking $x(1)\in\mathbb{Z}$.
This  can be used to store the sheet on which we evaluate multivalued functions arising from loop integrals.
Iterating this for all subgraphs $\Gamma_i$ which contain loops, we get a hierarchy of markings $x(i)$ for all edges not in $T$.
By  Thm.(\ref{homology}) we have $x(i)\in \mathbb{Z}$. If all cuts $\epsilon_k$ are two-particle cuts ($e_{\Gamma_{F_i}}-e_{\Gamma_{F_{i-1}}}=2=e_{\Gamma_{F_2}}$), every such edge has a different marking $x(i)$, $i=1,\ldots,|\Gamma|$. Else, the marking is degenerate. A comparison to the familiar marking by generators of the fundamental group of a graph (as in Outer Space) will be given in future work. The marking used here stores the homology and thus multi-valuedness assigned to functions $\Phi_R(\Gamma_F)$, $\Phi_R(\Gamma_i)$. This information is implicitly stored in $(M_i^\Gamma)_{rs}$. 
\end{remark}

Let us now indicate how the matrix looks in the case of a $4\times 4$ matrix.
We also indicate how Cutkosky's thorem and dispersion act between the matrix entries. 
The spanning tree $T$ has three edges, and spanning forests are given by indicating the edges to be removed.

Here is $M^\Gamma_i  =$
\[
\left(
\begin{array}{ccccccc}
1 & & 0 & & 0  & & 0\\
\uparrow\pi & & \uparrow\pi & & & & \\
\Upsilon^{T_2}_{\Gamma_2} & \rightleftharpoons_\mathrm{disp}^\mathrm{Var} & \Upsilon^{T_2-e_1}_{\Gamma_2} & & 0 & & 0\\
\uparrow\pi & & \uparrow\pi & & \uparrow\pi & & \\
\Upsilon^{T_3}_{\Gamma_3} & \rightleftharpoons_\mathrm{disp}^\mathrm{Var} & \Upsilon^{T_3-e_1}_{\Gamma_3}  & \rightleftharpoons_\mathrm{disp}^\mathrm{Var} & \Upsilon^{T_3-e_1-e_2}_{\Gamma_3} & & 0\\
\uparrow\pi & & \uparrow\pi & & \uparrow\pi & & \uparrow\pi \\
\Upsilon^{T_4=T}_{\Gamma_4=\Gamma} & \rightleftharpoons_\mathrm{disp}^\mathrm{Var} & \Upsilon^{T_4-e_1}_{\Gamma_4} & \rightleftharpoons_\mathrm{disp}^\mathrm{Var} & \Upsilon^{T_4-e_1-e_2}_{\Gamma_4} & \rightleftharpoons_\mathrm{disp}^\mathrm{Var} & \Upsilon^{T_4-e_1-e_2-e_3}_{\Gamma_4}
\end{array}\right).
\]
By construction, along the diagonal Cutkosky's thorem gives the variation associated to leading singularities.
This plus dispersion provided by a Hilbert transform wrt the channel variable defined by a suitable multiple edge $b_j, \Gamma_{F_i}/b_j=\Gamma_{F_{i-1}}$,
with $j=e_{\Gamma_{F_i}}-e_{\Gamma_{F_{i-1}}}$ 
determines the first subdiagonal.
Having thus determined the first subdiagonal, we can iterate to determine the whole matrix.

So in a row, we go to the right by computing the variation with the help of Cutkosky's theorem. We go to the left by the Hilbert transform. The entry from above is needed to fix the normalization in this dispersion integral.

Vertical arrows correspond to a
projection which gives the integrals over uncut edges as a fibration. 

At the end, starting from variations for leading thresholds, this systematic decomposition of graphs allows to reconstruct a graph from its variations through Hilbert transforms and the determination of anomalous thresholds.

We emphasize that a given graph $\Gamma$ gives many such matrices, one for each possibility to shrink its internal edges until all vertices unify in a single vertex. To understand the physical thresholds of a graph requires to analyse 
all these matrices.

\section{Dispersion}
Cutkosky's theorem allows us to find the variation associated to physical thresholds.
This variation can be associated to a pair $(\Gamma,T)$ with associated channel variable $s$. 

If there is a lowest finite real threshold $s_k>-\infty$, then for  $s<s_k$, we have real analyticity in this channel, with all other kinematical variables fixed:
\be
\Phi_R(\Gamma)(s)=\Phi_R^\star(\Gamma)(s^\star),\, s<s_k.
\ee
We conclude that for $s>s_k$ near the real axis
\be 
\Re(\Phi_R(\Gamma)(s))=\Re(\Phi_R(\Gamma)(s^\star))\ee
 and
\be 
\Im(\Phi_R(\Gamma)(s)))=-\Im(\Phi_R(\Gamma)(s^\star)).
\ee
It follows that the dispersion integral, when it exists, is an integral of the variation from $s_k$ to $+\infty$ along the real axis.

Depending on the superficial degree of divergence of a graph the form to be integrated for renormalized Feynman amplitudes has the structure
\be
\frac{\Phi_\Gamma^k\ln\frac{\Phi_\Gamma}{\Phi^0_\Gamma}}{\psi^j_\Gamma}\Omega_\Gamma, k\geq 0,
\ee
or
\be
\frac{\Phi_\Gamma^k}{\psi^j_\Gamma}\Omega_\Gamma, k<0.
\ee
We focus on the case $\frac{\ln\frac{\Phi_\Gamma}{\Phi^0_\Gamma}}{\psi^2_\Gamma}\Omega_\Gamma$, the other cases being similar.

The crucial input is the ratio $\Phi_\Gamma/\Phi^0_\Gamma$ of second Symanzik polynomials as  functions on $Q_\Gamma$ and the chosen renormalization point $Q_0\in Q_\Gamma$ which we take Euclidean or at least such that  $\Phi_\Gamma^0<0$ is negative definite.

Positivity of $\psi$ leaves the logarithm $\ln \Phi_\Gamma$ as the source of imaginary parts. We hence focus on the domain where its argument is smaller than zero.

Using the Heavyside $\Theta$-function ($\Theta(x)=1,x>0,\Theta(x)=0,x<0$), we get an imaginary part from a contribution
\[
\int_{\mathbb{P}^\Gamma} \frac{\Theta(\Phi_\Gamma)}{\psi_\gamma^2}\Omega_\Gamma.
\] 
\begin{remark}
This argument can be generalized in a straightforward manner to cover the case of a finite forest sum over graphs, or other degrees of superficial divergence than logarithmic. The crucial fact that the the imaginary part
of a graph amplitude is supported on the zero locus of the second Symanzik polynomial remains valid.  
\end{remark}

We use $\psi_\Gamma=\psi_{\Gamma/e}+t_\gamma R^\Gamma_\gamma$, for $\gamma$ some multiple edge $b_k$\footnote{If $\gamma$ is a 1-edge $b_1$, $t_\gamma=A_e, R^\Gamma_e=\psi(\Gamma-e)$}.
Here and in the following $\mathbb{P}^\Gamma$ indicates a positive real  integration chain given as the interior of the simplex $\sigma_\Gamma$ (see (\ref{simplex})).
 
A partial integration 
leaves us to consider $Var(\Gamma)=U_\Gamma^\emptyset+U_{\Gamma/\gamma}^\gamma$, with
\[
U_\Gamma^\emptyset:=\int_{\mathbb{P}^{e_{\Gamma}-2}}\int_0^\infty 
\frac{\Theta^\prime(\Phi_\Gamma)}{\psi_\gamma R^\Gamma_e}
dt_\gamma\wedge \Omega_{\Gamma/\gamma}\wedge\Omega_\gamma.
\] 
The boundary term is
\[
U_{\Gamma/\gamma}^\gamma=\int_{\mathbb{P}^{e_{\Gamma}-2}}\frac{\Theta(\Phi_{\Gamma/\gamma})}{\psi_{\Gamma/\gamma}R^\Gamma_\gamma}\Omega_{\Gamma/\gamma}\wedge \Omega_\gamma.
\] 
\begin{remark}
These boundary terms iterate along the partial integration of the first Symanzik polynomial $\psi$. Linear reducibility for this polynomial 
\cite{Brown,Panzer} comes in crucially. Note that minors of graphs $\Gamma$ are precisely what populates the boundary cells in the cubical chain complex.
\end{remark}

We have $\Phi_\Gamma=-Z(t_\gamma-t_1)(t_\gamma-t_2)$, with $t_{1/2}=Y\pm\sqrt{D}$
and $D=Y^2+XZ$ as before.

We find
\be
U^\emptyset_\Gamma=
\int_{\mathbb{P}^{e_{\Gamma}-2}}\int_0^\infty
\sum_{i=1}^2  \frac{\Phi^\prime_\Gamma\delta(t_\gamma-t_i)}{|\Phi^\prime_\Gamma(t_\gamma=t_i)|\psi_\Gamma R^\Gamma_\gamma}
dt_\gamma\wedge \Omega_{\Gamma/\gamma}\wedge\Omega_\gamma 
\ee
and therefore
\beas
U^\emptyset_\Gamma & = &  \int_{\mathbb{P}^{e_{\Gamma}-2}}\int_0^\infty
 \Theta(t_1)\overbrace{ \frac{\Phi^\prime_\Gamma(t_\gamma=t_1)}{|\Phi^\prime_\Gamma(t_\gamma=t_1)|(\psi_{\Gamma/e}+t_1 R^\Gamma_e) R^\Gamma_e}}^{=:\omega(t_1)}dt_\gamma\wedge \Omega_{\Gamma/\gamma}\wedge\Omega_\gamma 
\\
 & &  + \int_{\mathbb{P}^{e_{\Gamma}-2}}\int_0^\infty \Theta(t_2) \overbrace{\frac{\Phi^\prime_\Gamma(t_\gamma=t_2)}{|\Phi^\prime_\Gamma(t_\gamma=t_2)|(\psi_{\Gamma/e}+t_2 R^\Gamma_e) R^\Gamma_e}}^{=:\omega(t_2)}
 dt_\gamma\wedge \Omega_{\Gamma/\gamma}\wedge\Omega_\gamma 
. 
\eeas

Now $\Phi_\Gamma(t_\gamma=t_1)=+\sqrt{D}$ and
$\Phi_\Gamma(t_\gamma=t_2)=-\sqrt{D}$.

In the region were both $t_i>0$ we hence find
\[
U^\emptyset_\Gamma=
2 \int_{\mathbb{P}^{e_{\Gamma}-2}} \frac{\sqrt{D}}{(\psi_{\Gamma/e}+YR^\Gamma_e)^2-D(R^\Gamma_e)^2}
\Omega_{\Gamma/\gamma}\wedge\Omega_\gamma.
\] 
If only $t_1>0$
we find
\[
U^\emptyset_\Gamma  =    \int_{\mathbb{P}^{e_{\Gamma}-2}} 
\frac{\sqrt{D}}{|\Phi^\prime_\Gamma(A_e=t_1)|(\psi_{\Gamma/e}+t_1 R^\Gamma_e) R^\Gamma_e}
\Omega_{\Gamma/\gamma}\wedge\Omega_\gamma.
\]

Next, define the two domains $\omega^D_s,\omega^N_s\subseteq \mathbb{P}^{\Gamma/e}$ by
\beas
\omega^D_s & = & \{p\in \mathbb{P}^{\Gamma/e}| D(p,s)>0 \},\\
\omega^N_s & = & \{p\in \mathbb{P}^{\Gamma/e}| 4XZ(p,s)>0 \}.\\
\eeas
Here, $D$ and $4XZ$ are regarded as functions of the edge variables $A_e$ and the channel $s$, for all other kinematical variables and masses fixed and given.
Note that for $4XZ>0,D>0$ we have $Y<\sqrt{D}$ and for $4XZ<0,D>0$ we have $Y>\sqrt{D}$.

Then, $U^\emptyset_\Gamma=U_1^\emptyset(s)+U_2^\emptyset(s)$ with
\be
U_1^\emptyset(s):=\int_{\omega^D_s\cap \omega^N_s} \omega(t_1), 
\ee
and
\be
U_2^\emptyset(s):=\int_{\omega^D_s-(\omega^D_s\cap \omega^N_s)} (\omega(t_1)+\omega(t_2)). 
\ee
The dispersion integral is then obtained by bringing in all boundary terms iteratively
\beas
\Re(\Phi_R(\Gamma)) & = & \int_{s_k}^{s_{k-1}} \frac{U_2^\emptyset(s)}{(x-s)}ds
+\int_{s_0}^{\infty} \frac{U_1^\emptyset(s)}{(x-s)}ds\\
 & & +
\int_{s_{k-1}}^{s_{k-2}} \frac{U_2^\gamma(s)}{(x-s)}ds
+\int_{s_{k-1}}^{\infty} \frac{U_1^\gamma(s)}{(x-s)}ds\\
 & & +\cdots\\
  & & +\int_{s_{1}}^{\infty} \frac{U_1^{\gamma,\cdots}(s)}{(x-s)}ds.\\
\eeas

\section{Examples}
\subsection{One-loop bubble}
For the bubble, a double edge $\gamma=b_2(a,b)$ on two edges $e_1,e_2$ between two vertices $a,b$ with momenta $p_a=-p_b$
and channel $s=p_a^2$ with $\epsilon_2$ the cut separating $a$ and $b$ the only cut, 
there are two single lower triangle two-by-two matrix $M^\gamma_i$, $i\in 1,2$, depending on which of the two edges we take as the spanning tree $t$. 

For example
$$
M^\gamma_1=\bubblematrix
$$
and the three entries populate a cubical complex which is simply an interval $[0,1]$,
with the endpoint at zero associated to $(M^\gamma_1)_{11}$, the endpoint at $1$ to 
$(M^\gamma_1)_{22}$ and the interval $]0,1[$ to $(M^\gamma_1)_{21}$.

The second Symanzik polynomial is
\be
\Phi_\gamma(s,m_1^2,m_2^2)=s^2A_1A_2-(m_1^2A_1+m_2^2A_2)(A_1+A_2). 
\ee 
We also have $\psi_\gamma=A_1+A_2$, $\mathbb{P}_\gamma=\mathbb{P}^1$, $\Omega_\gamma=
A_1dA_2-A_2dA_1$ and the renormalized integral is
\be\label{bubbleint}
\Phi_R(\gamma)(s,s_0,m_1^2,m_2^2)=\int_{\mathbb{P}_\gamma} 
\frac{\ln \frac{\Phi(s,m_1^2,m_2^2)}{\Phi(s_0,m_1^2,m_2^2)}}{(A_1+A_2)^2}
\Omega_\gamma .
\ee
We choose $s_0$ such that $\Phi(s_0,m_1^2,m_2^2)<0$ for all $A_1,A_2>0$.
This means $s_0<(m_1+m_2)^2$. In particular, we distinguish the two matrices $M^\gamma_i$ my choosing either $m_1^2$ or $m_2^2$ for the renormalization point $s_0$:
\be
\Phi_R^1=\Phi_R(s,m_1^2,m_1^2,m_2^2),\,\Phi_R^2=\Phi_R(s,m_2^2,m_1^2,m_2^2).
\ee

The integral in Eq.(\ref{bubbleint})  can be easily integrated. Instead of doing it projectively, we can set say $A_1=1$ and integrate $A_2$ over the positive reals. 

Note that the second Symanzik polynomial is then quadratic in $A_2$, and its discriminant $\lambda$ determines the vanishing cell to be the real strip defined by the discriminant, \[-\sqrt{\lambda}\to +\sqrt{\lambda}.\]
We find
\be\label{bubble}
\Phi_R(s,s_0,m_1^2,m_2^2)= \frac{1}{4\pi^2}\frac{\sqrt{\lambda}}{2s}
\ln\frac{x+\sqrt{\lambda}}{x-\sqrt{\lambda}}-\frac{m_1^2-m_2^2}{2s}\ln\frac{m_1^2}{m_2^2}- (s\to s_0),
\ee
where $\lambda=\lambda(s,m_1^2,m_2^2)$ and $\lambda(a,b,c)\equiv a^2+b^2+c^2-2(ab+bc+ca)$ is the K\r{a}llen function, and $x=m_1^2+m_2^2-s$.

$\Phi_R(s,s_0,m_1^2,m_2^2)\equiv\Upsilon^{(a,b)}_\gamma$ (see (\ref{Upsilon}) and note that vertices $a,b$ form a 2-forest for the spanning  tree $T$ given by one of the two edges of the graph $\gamma$) has a variation
\bea\label{bubbleimag}
\Upsilon^{(a,b)}_\gamma\equiv Var(\gamma)\equiv\Delta_2(s) & = &\nonumber\\
 & &  2 i \frac{\sqrt{\lambda}}{4\pi s}\Theta(s-(m_1+m_2)^2)\\
 & = & \int d^4 k \delta^+(k^2-m_1^2)\delta^+((k+q)^2-m_2^2)\nonumber. 
\eea
It hence allows for a dispersion 
\be
 \Phi_R(\gamma)(s,s_0,m_1^2,m_2^2)=\frac{s-s_0}{\pi}\int_{(m_1+m_2)^2}^\infty
\frac{\Delta_2(x)}{(x-s)(x-s_0)}dx.\ee
\begin{remark}
Whilst the real part in (\ref{bubble}) exists for $s=0$,   (\ref{bubbleimag}) immediately shows that there is a singularity at $s=0$, which is a secondary
singularity, whose presence indicates a (mass-independent)  pinching at infinity. Note that $s=\sqrt{\lambda(s,0,0)}$.
\end{remark}
We find our first two matrices $M^1_\gamma,M^2_\gamma$:
\be\label{matrixbubble}
M^1_\gamma=\left(
\begin{array}{cc}
1 &  0\\
\Phi_R^1 & \Delta_2(s) 
\end{array}\right),\,M^2_\gamma=\left(
\begin{array}{cc}
1 &  0\\
\Phi_R^2 & \Delta_2(s) 
\end{array}\right).
\ee
Here we use $\Upsilon^{(a,b)}_\One=1$.
\subsection{The triangle}
We consider the triangle graph $\Delta$. 
In fact, we augment it with one of its three possible spanning trees, say on edges $e_2,e_3$, so $E_T=\{e_2,e_3\}$. The corresponding cell in the cubical chain complex is
\be \label{trianglecell}
\trianglecubical
\ee

For the  Cutkosky cut we choose two of the three edges, say $\epsilon_2=\{e_1,e_2\}$. This defines the channel $s=p_a^2$ and the matrix $M^\Delta_1$.

The  other cut in that matrix is the full cut  separating all three vertices.

We give $M^\Delta_1$ in the following figure:
$$
M^\Delta_1=\trianglematrix
$$
We now calculate:
\beas
\Phi_\Delta & = & \overbrace{p_a^2A_1A_2-(m_1^2A_1+m_2^2A_1)(A_1+A_2)}^{=\Phi_{\Gamma/e_3}}\\
 & & +A_3((p_b^2-m_3^2-m_1^2)A_1+(p_c^2-m_1^2-m_3^2)A_2)-A_3^2m_3^2,
\eeas
so 
\beas 
\Phi_\Delta & = & \Phi_{\Delta/e_3}
  +A_3\Phi^{m_3^2}_{\Delta-e_3}-A_3^2 m_3^2\overbrace{\psi_{\Delta-e_1}}^{=1},
\eeas
as announced ($A_3=t_\gamma$):
\[
X=\Phi_{\Delta/e_3},\,Y=\overbrace{(p_b^2-m_3^2-m_1^2)}^{=:l_1}A_1+\overbrace{(p_c^2-m_1^2-m_3^2)}^{=:l_2}A_2,\,Z=m_3^2.
\]
We have $Y_0=m_2 l_1+m_1 l_2$, and need $Y_0>0$ for a Landau singularity.

Solving $\Phi(\Delta/e_3)=0$ for a Landau singularity determines
the familiar physical threshold in the $s=p_a^2$ channel, leading for the reduced graph to
\be
p_Q: s_0=(m_2+m_3)^2,\, p_A: A_1m_1=A_2m_2. 
\ee

We let $D=Y^2+4XZ$ be the discriminant. For a Landau singularity we need 
\[
D=0.
\]
We have
\be
\Phi_\Delta=-m_3^2\left(A_3-\frac{Y+\sqrt{D}}{2m_3^2}\right)\left(A_3-\frac{Y-\sqrt{D}}{2m_3^2}\right),
\ee
where $Y,D$ are functions of $A_1,A_2$ and $m_1^2,m_2^2,m_3^2,s,p_b^2,p_c^2$.

We can write
\[
0=D=Y^2+4Z(sA_1A_2-N),
\] 
with $N=(A_1m_1^2+A_2m_2^2)(A_1+A_2)$ $s$-independent.

This gives
\be 
s(A_1,A_2)=\frac{4ZN-(A_1l_1+A_2l_2)^2}{4ZA_1A_2}=:\frac{A_1}{A_2}\rho_1+\rho_0+\frac{A_2}{A_1}\rho_2.
\ee

Define two Kallen functions $\lambda_1=\lambda(p_b^2,m_1^2,m_3^2)$ and
$\lambda_2=\lambda(p_c^2,m_2^2,m_3^2)$. Both are real and non-zero off their threshold or pseudo-threshold.

Then, for 
\[
r:=\lambda_1/\lambda_2>0,
\]
we find the threshold $s_1$ at
\be 
s_1=\frac{4m_3^2(\sqrt{\lambda_1}m_1^2+\sqrt{\lambda_2}m_2^2)(\sqrt{\lambda_1}+\sqrt{\lambda_2})-(\sqrt{\lambda_1}l_2+\sqrt{\lambda_2}l_1)^2}{4m_3^2\sqrt{\lambda_1}\sqrt{\lambda_2}}.
\ee
This can also be written as 
\be 
s_1=(m_1+m_2)^2+\frac{4m_3^2(\sqrt{\lambda_2}m_1-\sqrt{\lambda_1}m_2)^2-(\sqrt{\lambda_1}l_2+\sqrt{\lambda_2}l_1)^2}{4m_3^2\sqrt{\lambda_1}\sqrt{\lambda_2}}.
\ee

On the other hand for $r<0$ and therefore the coefficients of $\rho_1$, $\rho_2$ above of different sign we find a minimum
\be s_1=-\infty,
\ee
along either $A_1=0$ or $A_2=0$.

The domains $\omega^D_s,\omega^N_s$ can be easily determined from the above and determine the full complexity of the triangle function. 

\begin{remark}
Let us now discuss the triangle in more detail.
It allows three spanning trees on two edges each, so we get six matrices $M^\Delta_i$, $i=1,\ldots,6$ altogether,
by having two possibilities to order the two edges for each spanning tree.

The six matrices $M^\Delta_i$ come in groups of two for each spanning tree.

For each of the three spanning trees we get a cell as in (\ref{trianglecell}).

The boundary operator for such a cell in the cubical cell complex of \cite{VogtmannHatcher} is the obvious one stemming from co-dimension one hypersurfaces at $0$ or $1$
with suitable signs. So the square populated by the triangle $\Delta$ in (\ref{trianglecell}) has four boundary components, the edges populated by the four graphs as indicated. Those four edges are the obvious boundary of the square. 

If we now consider all graphs in (\ref{trianglecell}) as evaluated by the Feynman rules, we can consider for a given cell a boundary operator which replaces evaluation at 
the $x_e=0$-hypersurface by shrinking edge $e$, and evaluation at the $x_e=1$-hypersurface by setting edge $e$ on the mass-shell.

Then, to check that this is a boundary operator for the amplitudes defined by the graphs in (\ref{trianglecell}) we need to check that the amplitudes for the  four graphs at the four corners are uniquely defined from the amplitudes of the graphs at the adjacent edges: for example, the imaginary part of the amplitude of the  graph on the left vertical edge is related to the amplitude of the graph at the upper left corner:
This imaginary part must be also obtained from  shinking  edge $e_3$ in the graph on the upper horizontal edge by setting $A_3$ to zero in the integrand and integrating over the hypersurface $A_3=0$ of $\sigma_\Delta$. This is indeed the case, and similar checks work for all other corners.

A proof that we have a cubical chain complex on the level of amplitudes and a detailed study will be given in future work.  
\end{remark}

\end{document}